\documentclass[11pt]{article}
\usepackage{graphicx} 
\usepackage{float}

\usepackage[utf8]{inputenc}
\usepackage{setspace}

\usepackage[backend=biber,style=numeric,sorting=nty,
maxbibnames=99,doi=false,giveninits=true]{biblatex}
\usepackage{mhchem}
\usepackage{bbold}
\usepackage{dsfont}
\usepackage{mathtools}
\usepackage{stmaryrd}
\usepackage{verbatim}
\usepackage{amsthm}
\usepackage{amssymb}
\usepackage{titlesec}
\usepackage{mathrsfs}

\newcommand{\eps}{\varepsilon}

\newcommand{\PP}{\mathds{P}}
\newcommand{\R}{\mathds{R}}

\newcommand{\Z}{\mathds{Z}}
\newcommand{\E}{\mathds{E}} 

\newcommand{\I}{\mathcal{I}}

\newcommand{\J}{\mathcal{J}}

\newcommand{\CC}{\mathcal{C}}

\newcommand{\X}{\mathcal{X}}

\newcommand{\T}{\mathcal{T}}

\newcommand{\one}{\mathds{1}}

\usepackage{amsthm}
\newtheorem{theorem}{Theorem}[section]

\newtheorem{lemma}{Lemma}[section]

\newtheorem{remark}{Remark}[section]
\newtheorem{assumption}{Assumption}[section]

\numberwithin{equation}{section}

\theoremstyle{definition}

\usepackage{xcolor}
\usepackage[hidelinks]{hyperref}
\usepackage{xurl}
\hypersetup{
    colorlinks=true,
    linkcolor=blue,
    filecolor=blue,      
    urlcolor=blue,
    citecolor=blue,
    pdftitle={Overleaf Example},
    pdfpagemode=FullScreen,
}

\DeclarePairedDelimiterX{\floor}[1]{\lfloor}{\rfloor}{#1}
\DeclarePairedDelimiterX{\qvar}[1]{\langle}{\rangle}{#1}
\DeclarePairedDelimiterX{\inn}[1]{\langle}{\rangle}{#1}
\newcommand{\abs}[1]{\left\lvert #1 \right\rvert} 
\newcommand{\norm}[1]{\left\lVert #1 \right\rVert}

\newcommand{\Pn}{\nu}

\newcommand{\numI}{p}
\newcommand{\numJ}{\tilde{p}}
\newcommand{\Mo}{G}
\newcommand{\Wo}{H}

\newcommand{\FP}{\pi}
\newcommand{\QSS}{\eta}

\newcommand{\rr}[2]{\bar{r}_{#1}^{(#2)}}
\newcommand{\dd}[2]{\bar{d}_{#1}^{(#2)}}
\newcommand{\mi}[3]{\bar{\mu}_{#1}^{(#2,#3)}}
\newcommand{\size}[2]{\mathfrak{s}_{#1}^{(#2)}}
\newcommand{\fpone}[2]{\FP_{#1}^{(#2)}}
\newcommand{\fp}[2]{\rho_{#1}^{(#2)}}
\newcommand{\qss}[2]{\QSS_{#1}^{(#2)}}
\newcommand{\basisvec}[2]{e_{#1}^{(#2)}}

\newcommand{\sizenumeric}[2]{\mathfrak{s}_{#1}^{(#2)}}
\newcommand{\fponenumeric}[2]{\FP_{#1}^{(#2)}}
\newcommand{\fpnumeric}[2]{\rho_{#1}^{(#2)}}
\newcommand{\qssnumeric}[2]{\QSS_{#1}^{(#2)}}

\newcommand{\rrrate}[2]{r_{#1}^{(#2)}}
\newcommand{\ddrate}[2]{d_{#1}^{(#2)}}
\newcommand{\mirate}[3]{\mu_{#1}^{(#2,#3)}}

\usepackage[toc,page,titletoc]{appendix}
\usepackage[capitalise,nameinlink]{cleveref}

\usepackage[customcolors]{hf-tikz} 
\usetikzlibrary{patterns}
\usetikzlibrary{matrix,decorations.pathreplacing}
\pgfkeys{tikz/mymatrixenv/.style={decoration={brace},every left delimiter/.style={xshift=8pt},every right delimiter/.style={xshift=-8pt}}}
\pgfkeys{tikz/mymatrix/.style={matrix of math nodes,nodes in empty cells,left delimiter={[},right delimiter={]},inner sep=1pt,outer sep=1pt,column sep=3pt,row sep=3pt,nodes={minimum width=25pt,minimum height=13pt,anchor=center,inner sep=0pt,outer sep=0pt}}}
\pgfkeys{tikz/mymatrixbrace/.style={decorate,thick}}
\newcommand*\mymatrixbraceleftempty[4][m]{
	\draw[mymatrixbrace,color=white] (#1.east|-#1-#2-1.north east) -- node[right=2pt] {\textcolor{black}{#4}} (#1.east|-#1-#3-1.south east);
}
\newcommand*\mymatrixbraceleft[4][m]{
	\draw[mymatrixbrace] (#1.east|-#1-#2-1.north east) -- node[right=2pt] {#4} (#1.east|-#1-#3-1.south east);
}
\newcommand*\mymatrixbracetopempty[4][m]{
	\draw[mymatrixbrace,color=white] (#1.north-|#1-1-#2.north west) -- node[above=2pt] {\textcolor{black}{#4}} (#1.north-|#1-1-#3.north east);
}
\newcommand*\mymatrixbracetop[4][m]{
	\draw[mymatrixbrace] (#1.north-|#1-1-#2.north west) -- node[above=2pt] {#4} (#1.north-|#1-1-#3.north east);
}

\usepackage{authblk}
\title{Mutant Fixation for a Stochastic Evolutionary Model in Fragmented Populations\vspace{-0.5em}}

\author[1]{Yi Fu}
\author[1]{Natalia Komarova}
\affil[1]{Department of Mathematics, University of California, San Diego, 9500 Gilman Drive, La Jolla CA 92093-0112. Email: {\tt\small {yif064,nkomarova}@ucsd.edu}}
\date{}

\begin{document}

\maketitle

\vspace{-1cm}

\noindent \textbf{Abstract:}
Population fragmentation is a common feature of many biological systems. Understanding mutant fixation in such systems is challenging because the underlying stochastic dynamics are high-dimensional. In this work, we develop a general mathematical framework for analyzing stochastic evolution in fragmented populations connected by rare migration. The framework is sufficiently general to accommodate heterogeneous deme sizes, deme-dependent birth and death processes, and migration on arbitrary strongly connected directed networks with asymmetric migration rates. We show that, in the limit where migration occurs on a much slower timescale than within-deme dynamics, the full stochastic process can be reduced to a lower-dimensional Markov chain whose states correspond to configurations of fully mutant and fully wild-type demes. The reduction theorem establishes that fixation probabilities and absorption times of the original process are asymptotically determined by the corresponding quantities of a reduced chain. As an application, we derive explicit formulas for mutant fixation probabilities and fixation times in fragmented populations initiated by the introduction of a single mutant. The results provide a general and tractable approach for studying evolutionary dynamics in complex fragmented populations.

\noindent \textbf{Relevance to Life Sciences.}
Fragmentation is ubiquitous in biological populations and can strongly influence evolutionary outcomes. The framework developed here is applicable to a broad range of systems, including cancer cell populations evolving in spatially separated tissue compartments, bacterial colonies and biofilms connected by dispersal, and ecological populations inhabiting fragmented landscapes. The results quantify how local fixation within demes and migration between demes jointly determine the probability and timescale of mutant spread through the entire population. They also provide a general methodology for predicting evolutionary outcomes in populations whose connectivity is altered by environmental conditions or human activities, such as habitat fragmentation, urbanization, or conservation interventions.

\noindent \textbf{Mathematical Content.}
The paper develops a singular perturbation analysis for continuous time Markov chains describing fragmented populations. The dynamics are characterized by a separation of timescales between rapid within-deme evolution and slow migration events governed by a small parameter. Using asymptotic methods, we construct a reduced Markov chain on the space of deme-level absorbing configurations and prove convergence results for fixation probabilities and absorption times. The analysis yields leading-order asymptotic formulas for mutant fixation probabilities, mean absorption times, and conditional absorption times in the rare-migration limit. The framework is sufficiently general to accommodate a broad class of stochastic evolutionary models and provides a rigorous model-reduction approach for fragmented-population dynamics.  The stochastic singular perturbation theorems are general for studying the transient regime of linearly perturbed finite state continuous time Markov chains and have potential applications to other systems.



\begin{refsection}

\section{Introduction}

In population genetics, mutant fixation occurs when a mutant lineage introduced into a wild-type population eventually takes over the entire population. In this context, the key questions are:  how likely a mutant is to take over, and how long it takes. Fixation  of mutants in well-mixed populations has been studied in various contexts, starting from classical work of Haldane \cite{haldane1927mathematical}, Moran \cite{moran1958random}, and  Kimura \cite{kimura1962probability}. Starting from the stepping-stone model by Kimura and Weiss \cite{kimura1964stepping}, fixation probability and timing have been investigated extensively in the context of subdivided populations \cite{maruyama1970fixation, slatkin1981fixation},  lattice models \cite{durrett2008probability}, and graph-structured populations  \cite{bib:2005LHN, traulsen2009stochastic, allen2017evolutionary, hindersin2014counterintuitive}. 

It has been shown that the evolutionary dynamics  change significantly when the population is fragmented -- individuals mainly evolve within each deme via reproduction or death events, and demes are connected through occasional migration events. Marrec et al. \cite{bib:2021MLB} proposed a formal mapping to reduce each deme to a single entity and provide asymptotics of mutant fixation probabilities for certain structured populations. They showed how migration asymmetry changes the fixation of mutants in this system. Yagoobi and Traulsen \cite{yagoobi2021fixation} developed a coarse-grained description for fragmented populations where the dynamics in demes was described  by the Moran process, and demonstrated that the underlying graph structure can lead to amplification and suppression of selection.  Wodarz and Komarova \cite{bib:2025KW} proposed a coarse-grained approach  to  study probability of mutant fixation for mutants with proportionally scaled division and death rates (quasi-neutral mutants), showing that population fragmentation can  convert differences in turnover into a new effective fitness, even reversing the direction of selection.

In this work, we give a rigorous formulation of the coarse-graining for generic parameter values through stochastic singular perturbation analysis. We describe an evolutionary model in fragmented populations as a multi-dimensional birth-death process with migration, where we use a small parameter $\eps$ capturing the time scale separation between fast reproduction and death and slow migration. This model is described as a Stochastic Chemical Reaction Network (SCRN), which is a continuous time Markov chain. For a general introduction to SCRNs, we refer the reader to \cite{Anderson_Kurtz2015StochAnalysis,Wilkinson2006SMS}. We study the mutant fixation probabilities, the expected fixation time, and the conditional expected fixation time provided that the mutant fixates. 

Avrachenkov et al. \cite{Avrachenkov2013} studied singular perturbation for discrete time Markov chains with finite state space. Some of the results in \cite{Avrachenkov2013} were generalized to the continuous time Markov chain setting in Bruno et al. \cite{bib:epifinite}, in which the perturbed continuous time Markov chains are considered to be irreducible. 
Our present work also concerns continuous time Markov chains, while the perturbed chains are not irreducible. The method for characterizing absorbing probabilities used in  our present paper is an analog of what is in \cite {Avrachenkov2013} but for the continuous time setting. 
To characterize mean first passage times of the perturbed chain, Arvachenkov et al. \cite {Avrachenkov2013} first find series expansion for the deviation matrix of the perturbed chain and then use a formula to obtain mean first passage times. The expected exiting times that we study here are special cases, and we find a more direct way to characterize these mean first passage times, by using the same matrix inverse as used to characterize the absorbing probabilities. We also consider conditional expected exiting times, which was not studied in \cite {Avrachenkov2013}.
In addition, in the present study we provide interpretations to some key algebraic formulas in \cite {Avrachenkov2013}. For example, in our evolutionary models, a state in the reduced Markov chain is an aggregation of states in the original model with the same deme configuration. In this sense, states in the reduced chain can be intuitively interpreted to correspond to the deme configurations, which is consistent with the state definition in those formal coarse-graining approaches \cite{bib:2021MLB,yagoobi2021fixation,bib:2025KW}. 
Furthermore, we show that the leading coefficients of the series expansion of absorbing probabilities are closely related to the corresponding absorbing probabilities in the reduced Markov chain. We also have similar statements for expected exiting times and conditional expected exiting times.
The singular perturbation technique for continuous time Markov chains developed in this paper is general and has potential applications to systems other than this evolutionary model in fragmented populations as well.

\subsection{Notation}
Denote the set of non-negative integers by $\Z_+ = \{0,1,2, \ldots \}$. For an integer $d \geq 2$ we denote by $\Z_+^d$ the set of $d$-dimensional vectors with entries in $\Z_+$. 
The set of real numbers will be denoted by $\R$, and $d$-dimensional Euclidean space will be denoted by $\R^d$ for $d \geq 2$. 
We use $(\cdot)_+$ to mean that $(x)_+ = \max\{ 0,x\}$.
Vectors are column vectors unless indicated otherwise.
For integers $n,m \geq 1$, the set of $n \times m$ matrices with real-valued entries will be denoted by $\R^{n \times m}$. 
We denote by $\one$ a vector of any dimension where all entries are $1$'s. The size of $\one$ will be understood from the context. 
For a finite set $\X$, 
a matrix $Q = (Q_{x,y})_{x,y \in \X}$ is called an infinitesimal generator if $Q_{x,y} \geq 0$ for every $x \neq y \in \X$ and $Q\one = 0$. 
For a real-valued function $f$ of a real variable and an integer $n$, we say $f(\eps)=O(\eps^n)$ as $\eps \rightarrow 0$ if $\frac{f(\eps)}{\eps^n}$ is bounded for all $\eps$ in a sufficiently small neighborhood of zero.

\subsection{Structure of the paper}
The paper is organized as follows. We introduce the evolutionary model for fragmented population in
section \ref{sec:model}. In section \ref{sec:MotivatingExamples}, we use two motivating examples to introduce the mathematical setting and questions to address in this paper. The results presented in section \ref{sec:2D} are consequences of our key theorems in stochastic singular perturbation stated in section \ref{sec:SP}. These theorems are general to study the transient regime for singularly perturbed continuous time Markov chains with finite state space, and we apply these theorems to study the evolutionary model for fragmented population in section \ref{sec:general}. The proofs of these theorems are in \cref{sec:proofs}, and some additional details can be found in Supplementary Information (SI).

\section{Evolutionary Model for Fragmented Population}
\label{sec:model}

In this paper, we study fragmented population dynamics among $D$ demes. 
We present our evolutionary model for fragmented poplutation as a Stochastic Chemical Reaction Network (SCRN). More specifically, in each deme $\ell \in \{ 1, \dots, D \}$, there is a mixed population of wild-type species $W_\ell$ and mutant species $M_\ell$, where each individual can reproduce or die and the individuals compete under a shared carrying-capacity constraint $K_\ell$.
Migration events among demes happen at slower rates than reproduction or death. That is, individuals mainly evolve within each deme but demes are connected through occasional migration events. We use a small parameter $\eps$ to capture this time scale difference between the migration rate and the reproduction and death rates. We denote the number of species $W_\ell$ by $w_\ell$, and denote the number of species $M_\ell$ by $m_\ell$. 
Then, we can describe the dynamics of this system using the following reactions:
\begin{equation}
\begin{aligned}
\label{eqn:DemeReactionsMain}
    & W_\ell \ce{->[$\rrrate{w}{\ell} (x)$]} 2 W_\ell, \qquad  W_\ell \ce{->[$\ddrate{w}{\ell} (x)$]} \emptyset, \qquad W_\ell \ce{->[$\eps \mirate{w}{\ell}{k} (x)$]} W_{k}, \\
    & M_\ell \ce{->[$\rrrate{m}{\ell} (x)$]} 2 M_\ell, \qquad M_\ell \ce{->[$\ddrate{m}{\ell} (x)$]} \emptyset, \qquad M_\ell \ce{->[$\eps \mirate{m}{\ell}{k} (x)$]} M_{k},
\end{aligned}
\end{equation}
for $\ell \neq k \in \{ 1,2,\dots,D \}$, and over the arrow for these reactions, we write the reaction rates per capita, which depend on the current state $x=(w_1,m_1,w_2,m_2,\dots,w_D,m_D)$.

We will focus on systems with density-dependent birth and density-independent death\footnote{On the other hand, our theorems developed in section \ref{sec:SP} may be applied to the general version of the reaction network model \eqref{eqn:DemeReactionsMain}.}, 
and we assume that each deme does not undergo extinction. As a result, when there is only one individual in a deme, we set both the death rate and the migration rate (to other demes) for the deme to be zero. More explicitly, the reaction rates per capita in \eqref{eqn:DemeReactionsMain} are such that for $x=(w_1,m_1,w_2,m_2,\dots,w_D,m_D)$ and $\ell \neq k \in \{ 1,2,\dots,D \}$,
\begin{equation}
\label{eqn:ReactionRates}
\begin{aligned}
    \rrrate{w}{\ell} (x) = \rr{w}{\ell} \left(1 - \frac{w_\ell+m_\ell}{K_\ell} \right)_+, \quad & \ddrate{w}{\ell} (x) = \dd{w}{\ell} \one_{\{w_\ell+m_\ell>1\}}, & \mirate{w}{\ell}{k} (x) = \mi{w}{\ell}{k} \one_{ \left\{ \substack{w_\ell+m_\ell > 1, \\ w_k+m_k < K_k} \right\}}, \\
    \rrrate{m}{\ell} (x) = \rr{m}{\ell} \left(1 - \frac{w_\ell+m_\ell}{K_\ell} \right)_+ , \quad & \ddrate{m}{\ell} (x) = \dd{m}{\ell} \one_{\{w_\ell+m_\ell>1\}}, & \mirate{m}{\ell}{k} (x) = \mi{m}{\ell}{k} \one_{\left\{\substack{w_\ell+m_\ell > 1, \\ w_k+m_k < K_k} \right\}},
\end{aligned}
\end{equation}
where $\rr{w}{\ell}, \rr{m} {\ell}>0$ are  wild-type and mutant basic division rates, $\dd{w}{\ell}, \dd{m}{\ell} > 0$ are  wild-type and mutant  death rates, and $\mi{w}{\ell}{k}, \mi{m}{\ell}{k} \geq 0$ are wild-type and mutant migration rates.  Here, demes are connected through migration, but we do not require demes to be directly connected to all other demes.

For each $\eps \geq 0$, we let
\begin{equation}
\label{eqn:MC}
    X^{\eps} = \left\{ \left(X^\eps_{W_1} (t), X^\eps_{M_1} (t), X^\eps_{W_2} (t), X^\eps_{M_2} (t), \dots, X^\eps_{W_D} (t), X^\eps_{M_D} (t) \right) \in \Z_+^{2D}: t \geq 0 \right\}
\end{equation}
track the numbers of wild-type species $W_\ell$ and mutant species $M_\ell$ in each deme $\ell \in \{ 1, \dots, D \}$ in the system \eqref{eqn:DemeReactionsMain} over time $t$, and this $X^{\eps}$ is a continuous time Markov chain. The family of Markov chains $\{ X^\eps: \eps \geq 0 \}$ has a common finite state space
\begin{equation*}
    \X = \{(w_1,m_1,w_2,m_2,\dots,w_D,m_D): 1 \leq w_\ell+m_\ell \leq K_\ell  \text{ for each } \ell \in \{ 1,2,\dots,D\} \} \subseteq \Z_+^{2D}.
\end{equation*}
For $\ell \in \{ 1, \dots, D \}$, let $\basisvec{w}{\ell} = e_{2\ell-1}$ and $\basisvec{m}{\ell} = e_{2\ell}$, which are the standard unit basis vectors in $\R^{2D}$ for the coordinates designated for the wild-type in deme $\ell$ and the mutant in deme $\ell$, respectively. Then, the infinitesimal generator $Q(\eps)$ of $X^\eps$ is given by
\begin{equation}
\label{eqn:Qgen}
    Q_{x,x+v} (\eps) = 
    \begin{cases}
        \frac{\rr{w}{\ell}}{K_{\ell}} \left(K_{\ell} - (w_{\ell}+m_{\ell}) \right)  w_{\ell}, & \text{for } \ell \in \{ 1, \dots, D \} \text{ and } v = \basisvec{w}{\ell}, \\
        \frac{\rr{m}{\ell}}{K_{\ell}} \left(K_{\ell} - (w_{\ell}+m_{\ell}) \right)  m_{\ell}, & \text{for } \ell \in \{ 1, \dots, D \} \text{ and } v = \basisvec{m}{\ell}, \\
        \dd{w}{\ell} w_\ell & \text{for } \ell \in \{ 1, \dots, D \} \text{ and } v = - \basisvec{w}{\ell}, \\
        \dd{m}{\ell} m_\ell & \text{for } \ell \in \{ 1, \dots, D \} \text{ and } v = - \basisvec{m}{\ell}, \\
        \eps \mi{w}{\ell}{k} w_{\ell}, & \text{for } \ell \neq k \in \{ 1, \dots, D \} \text{ and } v=\basisvec{w}{k} - \basisvec{w}{\ell}, \\
        \eps \mi{m}{\ell}{k} m_{\ell}, & \text{for } \ell \neq k \in \{ 1, \dots, D \} \text{ and } v=\basisvec{m}{k} - \basisvec{m}{\ell}, \\
        0, & \text{for other } v \neq 0, \\
    \end{cases}
\end{equation}
for $x=(w_1,m_1,w_2,m_2,\dots,w_D,m_D) \in \X$ and $x+v \in \X$,
and $Q_{x,x} (\eps)$ is defined so that the row sums for $Q (\eps)$ are zeros.

\section{Motivating Examples}
\label{sec:MotivatingExamples}
Here we are concerned with the fate of a mutant introduced into a population initially composed of wild-type individuals. Of particular interest are the probability that the mutant ultimately replaces the wild-type and the timescale over which this fixation occurs. In section \ref{sec:1D}, we show how one may find probability of mutant fixation 
for system \eqref{eqn:DemeReactionsMain}--\eqref{eqn:ReactionRates} in a well mixed population (i.e., when $D=1$). As previously reported \cite{bib:2021MLB, bib:2025KW}, fragmentation can increase or decrease the probability of mutant fixation. In this paper, we develop a rigorous model reduction method (see section \ref{sec:SP}) that allows us to find the probability of mutant fixation in fragmented populations, as well as  the expected time until either the mutant or the wild-type fixates and the conditional expected time for fixation given that the mutant fixates. In section \ref{sec:2D}, we use the example of two demes of different sizes to present the type of results we can obtain from our theorems. Our theorems are very general in that they apply to any number of demes with different carrying capacities, different division and death rates, and different migration rates.  Our theorems have potential applications to other stochastic systems with time-scale separation as well.

\subsection{Probability of mutant fixation in a well-mixed populations}
\label{sec:1D}

Suppose there is a well-mixed population of wild-type species that has come to its stationary distribution. Then, a single mutant is introduced to this population through an immigration event from an external source. 
The immigration event can only take place if the population is below carrying capacity. Then, the wild-type and the mutant evolve according to \eqref{eqn:DemeReactionsMain}--\eqref{eqn:ReactionRates} with $D=1$, i.e.,
\begin{equation}
\label{eqn:CRN-WM}
\begin{aligned}
    W \ce{->[$\bar{r}_w (1 - (w+m)/K)_+$]} 2 W, \qquad W \ce{->[$\bar{d}_w \one_{\{ w+m > 1 \}}$]} \emptyset, \\
    M \ce{->[$\bar{r}_m (1 - (w+m)/K)_+$]} 2 M, \qquad M \ce{->[$\bar{d}_m \one_{\{ w+m > 1 \}}$]} \emptyset.
\end{aligned}
\end{equation}
where $\bar{r}_w$, $\bar{r}_m$, $\bar{d}_w$, $\bar{d}_m$ are positive constants, $w$ and $m$ are the number of wild-type individuals $W$ and mutant individuals $M$, respectively, and $K$ is the carrying capacity.
Let $X = \left\{ \left(X_{W} (t), X_{M} (t) \right) : t \geq 0 \right\}$ track the number of wild-type  and mutant individuals in the system \eqref{eqn:CRN-WM}. 
Here, since there are no migration events, this $X$ does not depend on $\eps$, and so we omit the $\eps$ in the notation in \eqref{eqn:MC}.
Then, $X$ is a continuous time Markov chain with state space $\X = \{(w,m) \in \Z_+^2: 1 \leq w+m \leq K \}$. 
It has two recurrent classes
\begin{equation*}
    \CC_w = \{(w,0) \in \X: 1 \leq w \leq K \} \qquad \text{ and } \qquad
    \CC_m = \{(0,m) \in \X: 1 \leq m \leq K \},   
\end{equation*}
and we shall use $\QSS_w$ and $\QSS_m$ to denote the stationary distributions for $X$ supported on $\CC_w$ and $\CC_m$, respectively. For $1 \leq n \leq K$, we use $\QSS_w (n)$ to denote the probability that $X$ is at the state $(n,0)$ under the stationary distribution $\QSS_w$, which depends on $\bar{r}_w,\bar{d}_w,K$. 
Similarly, for $1 \leq n \leq K$, we use $\QSS_m (n)$ to denote the probability that $X$ is at the state $(0,n)$ under the stationary distribution $\QSS_m$, which depends on $\bar{r}_m,\bar{d}_m,K$.
For $x = (n,k) \in \X$, we use $\FP_m (n,k)$ to denote the probability that $X$ reaches $\CC_m$ starting from $n$ wild-type individuals and $k$ mutant individuals, which depends on $\bar{r}_w,\bar{d}_w,\bar{r}_m,\bar{d}_m,K$. These probabilities $\FP_m (\cdot,\cdot)$ of mutant fixation given initial states can be calculated analytically by setting up a system of equations using first step analysis. 

When an immigration event introduces one mutant individual from an external source to the population of only the wild-type species that has come to its stationary distribution $\QSS_w$, the probability that the mutant takes over the whole population is given by
\begin{equation*}
    \rho_m = \sum_{n=1}^{K-1} \QSS_w (n) \FP_m (n,1).
\end{equation*}

\subsection{Wild-type and mutant dynamics for two demes}
\label{sec:2D}

As an illustrating example for analyzing models for fragmented populations, we consider the evolutionary model \eqref{eqn:DemeReactionsMain}--\eqref{eqn:ReactionRates} with $D=2$ and for $\ell \neq k \in \{ 1,2 \}$,
\begin{equation}
\label{eqn:2DrateConst}
    \rr{w}{\ell} = \bar{r}_w, \quad \rr{m}{\ell} = \bar{r}_m, \quad \dd{w}{\ell} = \bar{d}_w, \quad \dd{m}{\ell} = \bar{d}_m, \quad \mi{w}{\ell}{k} = \bar{\mu}_m, \quad \mi{m}{\ell}{k} = \bar{\mu}_m,
\end{equation}
where $\bar{r}_w,\bar{r}_m,\bar{d}_w,\bar{d}_m,\bar{\mu}_w,\bar{\mu}_m$ are positive constants that do not depend on deme indices. Here, the carrying capacities $K_1$ and $K_2$ for the two demes may or may not be the same.
The family $\{ X^\eps: \eps \geq 0 \}$ of SCRNs for this model has a common finite state space
\begin{equation*}
    \X = \{(w_1,m_1,w_2,m_2) \in \Z_+^4: 1 \leq w_1+m_1 \leq K_1, 1 \leq w_2+m_2 \leq K_2 \}.
\end{equation*}
The infinitesimal generator $Q(\eps)$ of $X^\eps$ is given by \eqref{eqn:Qgen} with \eqref{eqn:2DrateConst} and $D=2$.

Our method is based on using a reduced Markov chain, $\hat{X}_\mathcal{R}$, (see Figure \ref{fig:XhatI}c), which has four states $\hat{\X} = \{ ww,wm,mw,mm \}$ for this model. The state $wm$ in $\hat{\X}$, for example, is a condensation of the collection of states in $\X$ for which there are only wild-type individuals in the first deme and only mutant individuals in the second deme, i.e., the collection $\tilde{\CC}_{wm} = \{(w_1,0,0,m_2) \in \X: 1 \leq w_1 \leq K_1, 1 \leq m_2 \leq K_2\}$.
The infinitesimal generator $Q_\mathcal{R}$ for $\hat{X}_\mathcal{R}$ is given by
\begin{equation}
\label{eqn:QR2D}
\begin{aligned}
    (Q_\mathcal{R})_{mw,ww} = \bar{\mu}_w \sizenumeric{w}{2} \fpnumeric{w}{1}, \qquad (Q_\mathcal{R})_{mw,mm} = \bar{\mu}_m \sizenumeric{m}{1} \fpnumeric{m}{2}, \\
    (Q_\mathcal{R})_{wm,ww} = \bar{\mu}_w \sizenumeric{w}{1} \fpnumeric{w}{2}, \qquad (Q_\mathcal{R})_{wm,mm} = \bar{\mu}_m \sizenumeric{m}{2} \fpnumeric{m}{1},
\end{aligned}  
\end{equation}
where for $\ell \in \{ 1,2 \}$, $\size{m}{\ell}$ (resp. $\size{w}{\ell}$) is approximately the expected size of a well-mixed population of the mutant (resp. wild-type) species with carrying capacity $K_\ell$ under its stationary distribution, and $\fp{m}{\ell}$ (resp. $\fp{w}{\ell}$) is the probability that the mutant (resp. wild-type) eventually fixates starting from the stationary distribution for a well-mixed wild-type (resp. mutant) population with carrying capacity $K_\ell$ as one mutant (resp. wild-type) individual is introduced from outside deme $\ell$.
The precise definitions for these quantities are given in \eqref{eqn:srho1}--\eqref{eqn:srho2}.

Here we present the main results for this evolutionary model with two demes, while the derivation is deferred to section \ref{sec:general} (and SI - section \ref{sec:2Dresults}).
Assuming the wild-type population has come to equilibrium, we then consider an immigration event where a mutant is introduced from an external source and the probability it falls in a particular deme is proportional to the carrying capacity of that deme.
Applying Theorem \ref{thm:AP}, we obtain that the probability that both demes are fixated by the mutant is
\begin{equation}
\label{eqn:FP2D}
    \frac{K_1}{K_1+K_2} \fpnumeric{m}{1} \frac{\bar{\mu}_m \sizenumeric{m}{1} \fpnumeric{m}{2}}{\bar{\mu}_w \sizenumeric{w}{2} \fpnumeric{w}{1} + \bar{\mu}_m \sizenumeric{m}{1} \fpnumeric{m}{2}} + \frac{K_2}{K_1+K_2} \fpnumeric{m}{2} \frac{\bar{\mu}_m \sizenumeric{m}{2} \fpnumeric{m}{1}}{\bar{\mu}_w \sizenumeric{w}{1} \fpnumeric{w}{2} + \bar{\mu}_m \sizenumeric{m}{2} \fpnumeric{m}{1}} + O(\eps).
\end{equation}

\begin{figure}[H]
    \centering
    \includegraphics[width=\linewidth]{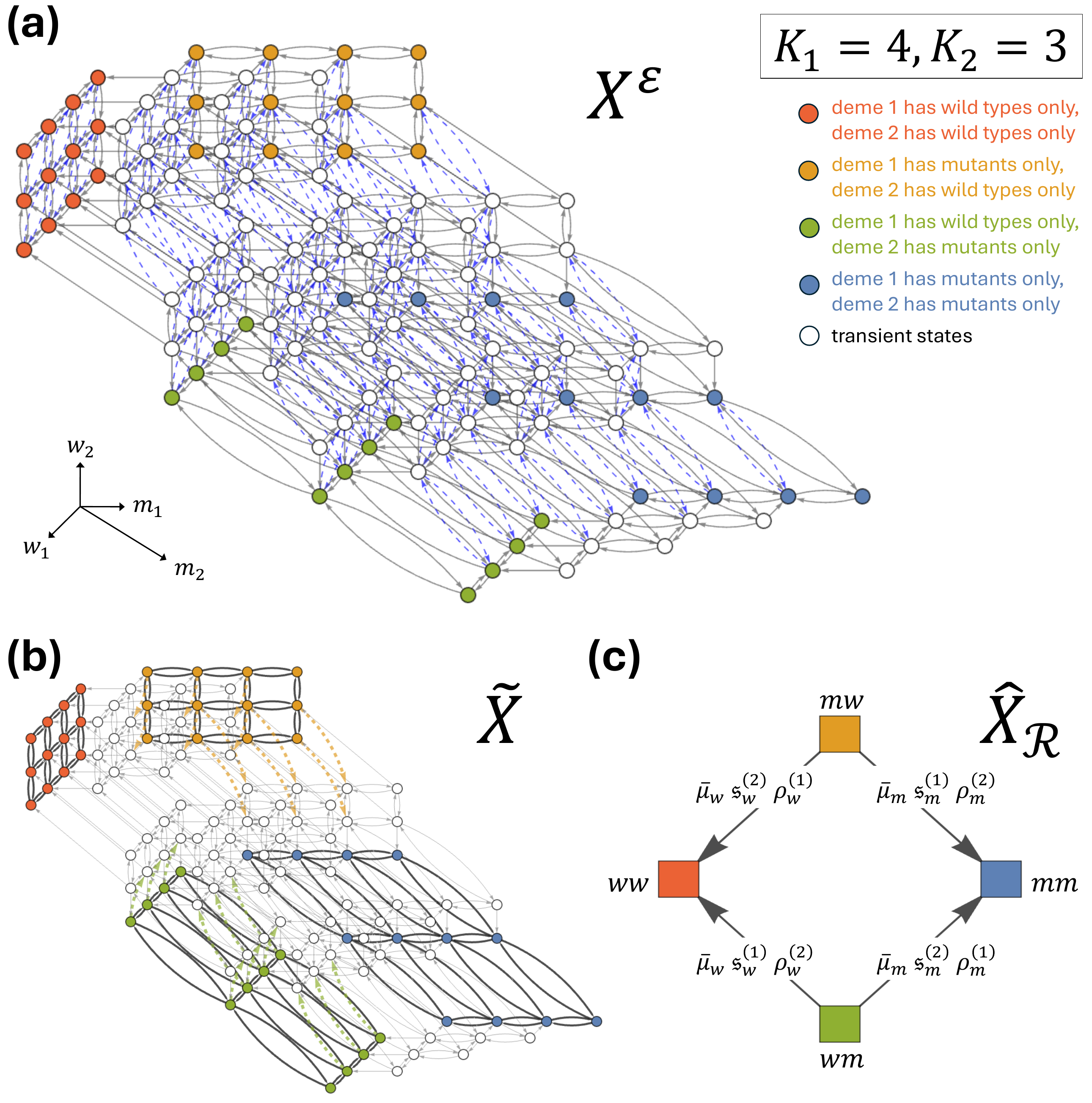}
    \caption{\textbf{Markov chain graphs for wild-type and mutant dynamics with two demes.} 
    \textbf{(a)} Markov chain graph for $X^\eps$ whose infinitesimal generator is given in \eqref{eqn:Qgen} with $D=2$, $K_1=4$ and $K_2=3$. We use gray arrows to represent possible one-step transitions with $O(1)$ rates, and blue dashed arrows to represent possible one-step transitions with $O(\eps)$ rates. 
    \textbf{(b)} From the Markov chain graph in (a), the transitions within each recurrent communicating class for $X^0$ is highlighted in black here, and the dashed arrows in (b) represents the transitions of $X^1$ that emanate from $\X \setminus \T$. The other dashed arrows in (a) are omitted here. Indeed, the Markov chain $\tilde{X}$ defined in Remark \ref{rmk} has Markov chain graph as in (b). If one traces the possible paths from a recurrent communicating class for $\tilde{X}$ through one dashed arrow and then gray arrows to reach another recurrent communicating class, then one will be able to find the possible one-step transitions in the reduced Markov chain $\hat{X}_\mathcal{R}$.
    \textbf{(c)} Markov chain graph for the reduced model $\hat{X}_\mathcal{R}$ whose infinitesimal generator is given in \eqref{eqn:QR2D}, which is consistent with the general expressions \eqref{eqn:QRk}--\eqref{eqn:QRl} for the $D$-deme evolutionary model.
    }
    \label{fig:XhatI}
\end{figure}

\noindent Applying Theorem \ref{thm:MFPT}, we obtain that the expected time until either the mutant or the wild-type fixates is
\begin{equation}
\label{eqn:MFPT2D}
    \left(\frac{K_1}{K_1+K_2}  \fpnumeric{m}{1} \frac{1}{\bar{\mu}_w \sizenumeric{w}{2} \fpnumeric{w}{1} + \bar{\mu}_m \sizenumeric{m}{1} \fpnumeric{m}{2}} + \frac{K_2}{K_1+K_2}  \fpnumeric{m}{2} \frac{1}{\bar{\mu}_w \sizenumeric{w}{1} \fpnumeric{w}{2} + \bar{\mu}_m \sizenumeric{m}{2} \fpnumeric{m}{1}} \right) \eps^{-1} + O(1).
\end{equation}
Applying Theorem \ref{thm:CMFPT}, we obtain that conditional expected time for fixation given that the mutant fixates is 
\begin{equation}
\label{eqn:CMFPT2D}
	\left( \frac{\mathcal{K}_1}{\mathcal{K}_1+\mathcal{K}_2} \frac{1}{\bar{\mu}_w \sizenumeric{w}{2} \fpnumeric{w}{1} + \bar{\mu}_m \sizenumeric{m}{1} \fpnumeric{m}{2}}  + \frac{\mathcal{K}_2}{\mathcal{K}_1+\mathcal{K}_2} \frac{1}{\bar{\mu}_w \sizenumeric{w}{1} \fpnumeric{w}{2} + \bar{\mu}_m \sizenumeric{m}{2} \fpnumeric{m}{1}} \right) \eps^{-1} + O(1),
\end{equation}
where $\frac{\mathcal{K}_1}{\mathcal{K}_2}= \frac{K_1}{K_2} \frac{\sizenumeric{m}{1}}{ \sizenumeric{m}{2}} \frac{\bar{\mu}_w \sizenumeric{w}{1} \fpnumeric{w}{2} + \bar{\mu}_m \sizenumeric{m}{2} \fpnumeric{m}{1} }{\bar{\mu}_w \sizenumeric{w}{2} \fpnumeric{w}{1} + \bar{\mu}_m \sizenumeric{m}{1} \fpnumeric{m}{2}}.$

\section{Singular Perturbation Analysis}
\label{sec:SP}

For some $\eps_0 > 0$, we consider a family $\{X^{\eps}:\: 0 \leq \eps < \eps_0\}$ of continuous time Markov chains with a common finite state space $\X \subseteq \Z^d_+$ with $|\X| > 1$ and whose infinitesimal transition rates depend on the parameter $\eps$.  Suppose that for each $0 \leq \eps < \eps_0$, $X^{\eps}$ has an infinitesimal generator $Q(\eps)$ where $\eps \mapsto Q(\eps)$ is a \textbf{linear perturbation} of $Q(0)$, i.e.,
\begin{equation}
\label{eqn:QmatrixLP}
Q(\eps) = Q^{(0)} + \eps Q^{(1)}.
\end{equation}
When $\eps>0$, the state space $\X$ for $X^{\eps}$ has a common decomposition\footnote{For $x \neq y \in \X$, we have $Q_{xy} (\eps) = Q^{(0)}_{xy} + \eps Q^{(1)}_{xy} \geq 0$ for $0 \leq \eps<\eps_0$, and so if $Q_{xy} (\eps)$ is positive for some $\eps>0$, then it will be positive for all $0<\eps<\eps_0$ (although $Q_{xy} (\eps_0)$ may be zero). Therefore, under the linear perturbation assumption, the communicating class structure for $X^\eps$, when $\eps > 0$, does not depend on $\eps$.}: 
\begin{equation}
\label{eqn:XepsDecomp}
\X = \left(\bigcup_{i \in \I} \CC_i \right) \cup \T
\end{equation} 
where the set $\I$ is non-empty and finite, the sets $\CC_i: i \in \I$ and $\T$ are mutually disjoint, each $\CC_{i}, i \in \I$ is a non-empty finite set of states in $\X$, which is a (positive) recurrent communicating class for $X^\eps$, and $\T$ is a finite set of transient states for $X^{\eps}$. 
We assume that $\T$ is non-empty.
In this paper, we are interested in studying the first time that $X^\eps$ enters $\bigcup_{i \in \I} \CC_i$ and the probability for $X^\eps$ to be absorbed in each $\CC_i$, starting from a state in $\T$.

When $\eps = 0$, in the absence of $O(\eps)$ transitions, there may be additional (positive) recurrent communicating classes emerging from $\CC_i, i \in \I$ and from $\T$, and states in $\CC_i, i \in \I$ could become transient. Since we are interested in studying the behavior of $X^\eps$ until the time it enters $\bigcup_{i \in \I} \CC_i$, we are not going to further decompose $\CC_i$ according to the communicating class structure for $X^0$ for $i \in \I$. We further decompose the set $\T$ in \eqref{eqn:XepsDecomp}:
\begin{equation*}
    \T = \left(\bigcup_{j \in \J} \tilde{\CC}_j \right) \cup \tilde{\T}.
\end{equation*}
where $\J$ is finite (possibly empty), the sets $\tilde{\CC}_j,j \in \J$ and $\tilde{\T}$ are mutually disjoint, each $\tilde{\CC}_j: j \in \J$ is a non-empty finite set of states in $\X$, which is a (positive) recurrent communicating class for $X^0$, and $\tilde{\T}$ is a finite (possibly empty) set of transient states for $X^0$.

We may and do relabel the states in $\X$ so that $Q^{(0)}$ and $Q^{(1)}$ in \eqref{eqn:QmatrixLP} have the following form:
\begin{equation}
\label{eqn:Q0}
Q^{(0)} = 
\begin{tikzpicture}[baseline={0ex},mymatrixenv]
\matrix [mymatrix,inner sep=5pt] (m)  
{
E_{11}^{(0)} & 0 & 0 & 0 & \cdots & 0 & 0 \\
0 & \ddots & 0 & \vdots & \ddots & \vdots & \vdots \\
0 & 0 & E_{\numI\numI}^{(0)} & 0 & \cdots & 0 & 0 \\
0 & \cdots & 0 & \tilde{E}_{11}^{(0)} & 0 & 0 & 0 \\
\vdots & \ddots & \vdots & 0 & \ddots & 0 & \vdots \\
0 & \cdots & 0 & 0 & 0 & \tilde{E}_{\numJ\numJ}^{(0)} & 0 \\
R_1^{(0)} & \cdots & R_\numI^{(0)} & \tilde{R}_1^{(0)} & \cdots & \tilde{R}_{\numJ}^{(0)} & \tilde{T}^{(0)} \\
};
\draw (-0.5,-2) -- (-0.5,2);
\draw[densely dashed] (2.5,-2) -- (2.5,2);
\draw (-3.3,0.35) -- (3.3,0.35);
\draw[densely dashed] (-3.3,-1.55) -- (3.3,-1.55);
\mymatrixbracetop{1}{1}{$\CC_1$}
\mymatrixbracetopempty{2}{2}{$\cdots$}
\mymatrixbracetop{3}{3}{$\CC_\numI$}
\mymatrixbracetop{4}{4}{$\tilde{\CC}_{1}$}
\mymatrixbracetopempty{5}{5}{$\cdots$}
\mymatrixbracetop{6}{6}{$\tilde{\CC}_{\numJ}$}
\mymatrixbracetop{7}{7}{$\tilde{\T}$}
\mymatrixbraceleft{1}{1}{$\CC_1$}
\mymatrixbraceleftempty{2}{2}{$\vdots$}
\mymatrixbraceleft{3}{3}{$\CC_\numI$}
\mymatrixbraceleft{4}{4}{$\tilde{\CC}_1$}
\mymatrixbraceleftempty{5}{5}{$\vdots$}
\mymatrixbraceleft{6}{6}{$\tilde{\CC}_{\numJ}$}
\mymatrixbraceleft{7}{7}{$\tilde{\T}$}
\end{tikzpicture}
=
\begin{tikzpicture}[baseline={0ex},mymatrixenv]
\matrix [mymatrix,inner sep=5pt] (m)  
{
E_{11}^{(0)} & 0 & 0 & 0\\
0 & \ddots & 0 & \vdots  \\
0 & 0 & E_{\numI\numI}^{(0)} & 0 \\
R_{0,1} & \dots & R_{0,\numI} & T_0\\
};
\draw (1,-1.2) -- (1,1.2);
\draw (-1.9,-0.7) -- (1.9,-0.7);
\mymatrixbracetop{1}{1}{$\CC_1$}
\mymatrixbracetopempty{2}{2}{$\cdots$}
\mymatrixbracetop{3}{3}{$\CC_\numI$}
\mymatrixbracetop{4}{4}{$\T$}
\mymatrixbraceleft{1}{1}{$\CC_1$}
\mymatrixbraceleftempty{2}{2}{$\vdots$}
\mymatrixbraceleft{3}{3}{$\CC_\numI$}
\mymatrixbraceleft{4}{4}{$\T$}
\end{tikzpicture}
\end{equation}
and
\begin{equation*}
Q^{(1)} = 
\begin{tikzpicture}[baseline={0ex},mymatrixenv]
\matrix [mymatrix,inner sep=5pt] (m)  
{
E_{11}^{(1)} & 0 & 0 & 0 & \cdots & 0 & 0 \\
0 & \ddots & 0 & \vdots & \ddots & \vdots & \vdots \\
0 & 0 & E_{\numI\numI}^{(1)} & 0 & \cdots & 0 & 0 \\
R_{11}^{(1)} & \cdots & R_{1 \numI}^{(1)} & \tilde{E}_{11}^{(1)} & \cdots & \tilde{E}_{1\numJ}^{(1)} & \tilde{U}_{1}^{(1)} \\
\vdots & \ddots & \vdots & \vdots & \ddots & \vdots & \vdots \\
R_{\numJ 1}^{(1)} & \cdots & R_{\numJ \numI}^{(1)} & \tilde{E}_{\numJ 1}^{(1)} & \cdots & \tilde{E}_{\numJ\numJ}^{(1)} & \tilde{U}_{\numJ}^{(1)} \\
R_1^{(1)} & \cdots & R_\numI^{(1)} & \tilde{R}_1^{(1)} & \cdots & \tilde{R}_{\numJ}^{(1)} & \tilde{T}^{(1)} \\
};
\draw (-0.5,-2) -- (-0.5,2);
\draw[densely dashed] (2.5,-2) -- (2.5,2);
\draw (-3.3,0.35) -- (3.3,0.35);
\draw[densely dashed] (-3.3,-1.55) -- (3.3,-1.55);
\mymatrixbracetop{1}{1}{$\CC_1$}
\mymatrixbracetopempty{2}{2}{$\cdots$}
\mymatrixbracetop{3}{3}{$\CC_\numI$}
\mymatrixbracetop{4}{4}{$\tilde{\CC}_1$}
\mymatrixbracetopempty{5}{5}{$\cdots$}
\mymatrixbracetop{6}{6}{$\tilde{\CC}_{\numJ}$}
\mymatrixbracetop{7}{7}{$\tilde{\T}$}
\mymatrixbraceleft{1}{1}{$\CC_1$}
\mymatrixbraceleftempty{2}{2}{$\vdots$}
\mymatrixbraceleft{3}{3}{$\CC_\numI$}
\mymatrixbraceleft{4}{4}{$\tilde{\CC}_1$}
\mymatrixbraceleftempty{5}{5}{$\vdots$}
\mymatrixbraceleft{6}{6}{$\tilde{\CC}_{\numJ}$}
\mymatrixbraceleft{7}{7}{$\tilde{\T}$}
\end{tikzpicture}
=
\begin{tikzpicture}[baseline={0ex},mymatrixenv]
\matrix [mymatrix,inner sep=5pt] (m)  
{
E_{11}^{(1)} & 0 & 0 & 0\\
0 & \ddots & 0 & \vdots  \\
0 & 0 & E_{\numI\numI}^{(1)} & 0 \\
R_{1,1} & \dots & R_{1,\numI} & T_1\\
};
\draw (1,-1.2) -- (1,1.2);
\draw (-1.9,-0.7) -- (1.9,-0.7);
\mymatrixbracetop{1}{1}{$\CC_1$}
\mymatrixbracetopempty{2}{2}{$\cdots$}
\mymatrixbracetop{3}{3}{$\CC_\numI$}
\mymatrixbracetop{4}{4}{$\T$}
\mymatrixbraceleft{1}{1}{$\CC_1$}
\mymatrixbraceleftempty{2}{2}{$\vdots$}
\mymatrixbraceleft{3}{3}{$\CC_\numI$}
\mymatrixbraceleft{4}{4}{$\T$}
\end{tikzpicture}
\end{equation*}
where $|\I|=\numI$ and $|\J|=\numJ$, and in the above block matrices (and in all block matrices in this paper), we use $\CC_1, ..., \CC_p$ to enumerate $\{\CC_i: i\in \I \}$ and use $\tilde{\CC}_1, ..., \tilde{\CC}_{\tilde{p}}$ to enumerate $\{\tilde{\CC}_j: j \in \J \}$.

Consider $0 < \eps < \eps_0$. Let $\tau^\eps = \inf \{ t \geq 0 : X^\eps (t) \in \X \setminus \T \}$. For $i \in \I$, denote the probability for $X^\eps$ to be absorbed in $\CC_i$, starting from $y \in \T$, by
\begin{equation}
\label{eqn:Lyi}
    \mathcal{L}_{yi} (\eps) = \PP_y^\eps \left[X^\eps (\tau^\eps) \in \CC_i \right].
\end{equation}
We denote the expected value of the first time that $X^\eps$ exits from $\T$, starting from $y \in \T$, by
\begin{equation}
\label{eqn:hy}
    h_y (\eps) = \E_y^\eps [\tau^\eps],
\end{equation}
and for $i \in \I$, denote the conditional expected value of the first time that $X^\eps$ exits from $\T$ given that $X^\eps$ is absorbed in $\CC_i$, starting from $y \in \T$, by
\begin{equation}
\label{eqn:gy}
    g_{yi} (\eps) = \E_y^\eps \left[ \tau^\eps | X^\eps (\tau^\eps) \in \CC_i \right].
\end{equation}
In Theorems \ref{thm:AP}--\ref{thm:CMFPT},
we characterize the series expansions of the quantities in \eqref{eqn:Lyi}--\eqref{eqn:gy} using a reduced Markov chain $\hat{X}_\mathcal{R}$ that we now introduce.

\subsection{Reduced Markov chain $\hat{X}_\mathcal{R}$}
\label{sec:reducedMC}

In this section, we will consider a continuous time Markov chain $\hat{X}_\mathcal{R}$ with state space $\mathcal{R} = \I \cup \J$, which helps explain the asymptotics of $\mathcal{L}_{yi} (\eps)$, $h_y (\eps)$ and $g_{yi} (\eps)$ in \eqref{eqn:Lyi}--\eqref{eqn:gy}. The infinitesimal generator $Q_{\mathcal{R}}$ for $\hat{X}_\mathcal{R}$ will be given in \eqref{eqn:QR}.

For each $j \in \J$, let $\tilde{\Pn}_j$ be the (unique) stationary distribution for the recurrent class $\tilde{\CC}_j$ in $X^0$ and $\tilde{\one}_j$ be the column vector of all ones with size $|\tilde{\CC}_j|$, and let $\tilde{L}_j = (\PP_y^0 [X^0(\tau^0) \in \tilde{\CC}_j] )_{y \in \tilde{\T}}$.
Then, let
\begin{equation}
\label{eqn:MW}
\Mo_0 = 
\begin{tikzpicture}[baseline={-0.5ex},mymatrixenv]
\matrix [mymatrix,inner sep=5pt] (m)  
{
\tilde{\Pn}_{1} & 0 & 0 & 0 \\
0 & \ddots & 0 & \vdots \\
0 & 0 & \tilde{\Pn}_{\numJ} & 0 \\
};
\draw[densely dashed] (1,-0.9) -- (1,0.9);
\mymatrixbracetop{1}{1}{$\tilde{\CC}_1$}
\mymatrixbracetopempty{2}{2}{$\cdots$}
\mymatrixbracetop{3}{3}{$\tilde{\CC}_{\numJ}$}
\mymatrixbracetop{4}{4}{$\tilde{\T}$}
\mymatrixbraceleft{1}{3}{$\J$}
\end{tikzpicture}
\quad \text{and} \quad 
\Wo_0 = 
\begin{tikzpicture}[baseline={-0.5ex},mymatrixenv]
\matrix [mymatrix,inner sep=5pt] (m)  
{
\tilde{\one}_{1} & 0 & 0 \\
0 & \ddots & 0 \\
0 & 0 & \tilde{\one}_{\numJ} \\
\tilde{L}_{1} & \cdots & \tilde{L}_{\numJ} \\
};
\draw[densely dashed] (-1.4,-0.6) -- (1.4,-0.6);
\mymatrixbracetop{1}{3}{$\J$}
\mymatrixbraceleft{1}{1}{$\tilde{\CC}_1$}
\mymatrixbraceleftempty{2}{2}{$\vdots$}
\mymatrixbraceleft{3}{3}{$\tilde{\CC}_{\numJ}$}
\mymatrixbraceleft{4}{4}{$\tilde{\T}$}
\end{tikzpicture}.
\end{equation}
Here, $\Mo_0 \in \R^{\abs{\J} \times \abs{\T}}$ and $\Wo_0 \in \R^{\abs{\T} \times \abs{\J}}$. Any vector obtained from augmenting a row of $\Mo_0$ by zeros corresponding to states in $\X \setminus \T$ is a stationary distribution for $X^0$, and for each $x \in \T$ and $j \in \J$, $(\Wo_0)_{x,j}$ gives the probability for $X^0$ to be absorbed in $\tilde{\CC}_j$, starting from $x$. In fact, the rows of $\Mo_0$ forms a basis for the null space of the transpose of $T_0$, and the columns of $\Wo_0$ forms a basis for the null space of $T_0$ (see Lemma \ref{lem:MWbases}). 
For each $i \in \I$, let $\one_i$ be the column vector of all ones with size $\abs{\CC_i}$, $L_i = \left(\PP_y^0 [X^0(\tau^0) \in \CC_i] \right)_{y \in \tilde{\T}}$ and
\begin{equation}
\label{eqn:Ji}
     J_i = \left(\tilde{\Pn}_j R^{(1)}_{ji} \one_i + \tilde{\Pn}_j U^{(1)}_j L_i \right)_{j \in \J}.
\end{equation} 
Now, we define
\begin{equation}
\label{eqn:QR}
Q_{\mathcal{R}} = 
\begin{tikzpicture}[baseline={-0.5ex},mymatrixenv]
\matrix [mymatrix,inner sep=5pt] (m)  
{
0 & \cdots & 0 & 0 \\
\vdots & \ddots & \vdots & \vdots \\
0 & \cdots & 0 & 0 \\
J_1 & \cdots & J_{\numI} & \Mo_0 T_1 \Wo_0 \\
};
\draw (0.7,-1.1) -- (0.7,1.1);
\draw (-2.1,-0.6) -- (2.1,-0.6);
\mymatrixbracetop{1}{3}{$\I$}
\mymatrixbracetop{4}{4}{$\J$}
\mymatrixbraceleft{1}{3}{$\I$}
\mymatrixbraceleft{4}{4}{$\J$}
\end{tikzpicture}.
\end{equation}
Then, $Q_{\mathcal{R}}$ is an infinitesimal generator (see Lemma \ref{lem:QR}).
In fact, for $j \in \J$ and $k \in \mathcal{R} \setminus \{ j \}$, 
\begin{equation*}
\left( Q_{\mathcal{R}} \right)_{jk} = 
\begin{cases}
    \tilde{\Pn}_j ( R_{jk}^{(1)} \one_k + \tilde{U}_{j}^{(1)} L_k ) & \text{ if } k \in \I, \\
    \tilde{\Pn}_j ( \tilde{E}_{jk}^{(1)} \tilde{\one}_k + \tilde{U}_{j}^{(1)} \tilde{L}_k ) & \text{ if } k \in \J,
\end{cases}
\end{equation*}
which is obtained by initiating a Markov chain at states in $\tilde{\CC}_j$ with distribution $\tilde{\Pn}_j$, where the infinitesimal rate of going to $\CC_k$ (or $\tilde{\CC}_k$) is given by the direct rate for $X^1$ plus the sum over $y \in \tilde{\T}$ of the infinitesimal rate of $X^1$ going to the transient state $y$ times the probability $(L_k)_y$ (or $(\tilde{L}_k)_y$) of $X^0$ then eventually going to $\CC_k$ (or $\tilde{\CC}_k$) from $y$.

\begin{remark}
    Suppose that each $\CC_i$ is a communicating class for $X^0$. For $x,y \in \X$, we let $\Pi^{(0)}_{xy} = \lim_{t \rightarrow \infty} \PP_x^0 [X^0(t) = y]$, and $\Pi^{(0)} = ( \Pi^{(0)}_{xy} )_{x,y \in \X}$. Then, there is a decomposition $\Pi^{(0)} = G^{(0)} H^{(0)}$, where rows of $G^{(0)} \in \R^{\abs{\I\cup\J} \times \abs{\X}}$ are given by stationary distributions for $X^0$ supported on each recurrent communicating class of $X^0$ and the entries in $H^{(0)}  \in \R^{\abs{\X} \times \abs{\I\cup\J}}$ are given by the probabilities that $X^0$ eventually is absorbed at each recurrent communicating class starting from each state. Indeed, $G_0$ and $H_0$ defined in \eqref{eqn:MW} are submatrices of $G^{(0)}$ and $H^{(0)}$, respectively. In addition, one can show that $Q_{\mathcal{R}} = G^{(0)} Q^{(1)} H^{(0)}$, and Avrachenkov et al. \cite{Avrachenkov2013} has a similar product of three matrices used to identify the generator for their aggregated chain in their context. Our matrix product $\Mo_0 T_1 \Wo_0$, which is a submatrix of $Q_{\mathcal{R}}$, is an analog of what Avrachenkov et al. \cite{Avrachenkov2013} used when characterizing the absorbing probabilities in their Section 6.2.2. On the other hand, when some $\CC_i$ can be further decomposed into multiple communicating classes for $X^0$, some parallel lines of reasoning remain valid.

\end{remark}

Let $\hat{X}_\mathcal{R}$ be a continuous time Markov chain with state space $\mathcal{R} = \I \cup \J$ and infinitesimal generator $Q_{\mathcal{R}}$. The states in $\I$ are absorbing for $\hat{X}_\mathcal{R}$, 
and we will assume the following throughout the paper.

\begin{assumption}
\label{ass:QRstructure}
    For $\hat{X}_\mathcal{R}$, each state in $\J$ is transient.
\end{assumption}

\begin{remark}
\label{rmk}
    One may check whether Assumption \ref{ass:QRstructure} holds by considering an auxiliary chain $\tilde{X}$ living in $\X$ whose transitions consist of the transitions of $X^0$ augmented by the transitions of $X^1$ that emanate from $\X \setminus \T$. For $j \in \J$, the state $j$ is transient for $\hat{X}_\mathcal{R}$ if and only if any state in $\CC_j$ is transient for this auxiliary chain. Note that states in $\CC_j$ communicate in the auxiliary chain. 
\end{remark}

Let $\hat{\tau} = \inf \{ t \geq 0 : \hat{X}_\mathcal{R} (t) \in \I \}$. 
For $i \in \I$, denote the probability for $\hat{X}_\mathcal{R}$ to be absorbed in state $i$, starting from $j \in \J$, by
\begin{equation}
\label{eqn:Lhatji}
    \hat{\mathcal{L}}_{ji} = \hat{\PP}_j \left[ \hat{X}_\mathcal{R} (\hat{\tau}) = i \right].
\end{equation}
We denote the expected value of the first time that $\hat{X}_\mathcal{R}$ exits from $\J$, 
starting from $j \in \J$, by
\begin{equation}
\label{eqn:hhatj}
    \hat{h}_j = \hat{\E}_j [\hat{\tau}],
\end{equation}
and for $i \in \I$, denote the conditional expected value of the first time that $\hat{X}_\mathcal{R}$ exits from $\J$ 
given that $\hat{X}_\mathcal{R}$ is absorbed in state $i$, starting from $j \in \J$, by
\begin{equation}
\label{eqn:ghatji}
    \hat{g}_{ji} = \hat{\E}_j \left[ \hat{\tau} \Big| \hat{X}_\mathcal{R} (\hat{\tau}) = i \right].
\end{equation}

\subsection{Key theorems on the asymptotic expansion}
\label{sec:mainResult}
Here we introduce our theorems for stochastic singular perturbation analysis, which characterize the leading terms of $\mathcal{L}_{yi} (\eps)$, $h_y (\eps)$ and $g_{yi} (\eps)$ in \eqref{eqn:Lyi}--\eqref{eqn:gy}. We will then apply these theorems to study the mutant fixation for fragmented population in section \ref{sec:general}.

\begin{theorem}
\label{thm:AP}
    Suppose Assumption \ref{ass:QRstructure} holds. Consider $i \in \I$ and $0 < \eps < \eps_0$. The vector $\mathcal{L}_{i} (\eps) = \left( \mathcal{L}_{yi} (\eps) \right)_{y \in \T}$ for the absorbing probabilities given in \eqref{eqn:Lyi} has a power series expansion.
    Moreover, for each $j \in \J$ and $y \in \tilde{\CC}_j$, 
    \begin{equation}
    \label{eqn:mainthmAPCj}
        \mathcal{L}_{yi} (\eps) = \hat{\mathcal{L}}_{ji} + O(\eps),
    \end{equation}
    and for each $y \in \tilde{\T}$,
    \begin{equation}
    \label{eqn:mainthmAP}
        \mathcal{L}_{yi} (\eps) = \left( \left( L_{i} \right)_y + \sum_{j \in \J} (\tilde{L}_j)_y \hat{\mathcal{L}}_{ji} \right) + O(\eps).
    \end{equation}
\end{theorem}

Theorem \ref{thm:AP} says that when $\eps > 0$ is small, the probability for $X^\eps$ to be absorbed in $\CC_i$ starting from $y \in \tilde{\CC}_j$ is approximately the probability for the reduced Markov chain $\hat{X}_\mathcal{R}$ to be absorbed at $i$ starting from $j$, and the probability for $X^\eps$ to be absorbed in $\CC_i$ starting from $y \in \tilde{\T}$ is approximately the probability for $X^0$ to directly go to $\CC_i$ plus 
the weighted sum of the probabilities for $\hat{X}_\mathcal{R}$ to be absorbed at $i$ starting from each state $j \in \J$, where the weights are given by the probability for $X^0$ to be absorbed in $\tilde{\CC}_j$ starting from $y$.

\begin{theorem}
\label{thm:MFPT}
    Suppose Assumption \ref{ass:QRstructure} holds. Consider $0 < \eps < \eps_0$. Suppose $\abs{\J} \geq 1$. The vector $h (\eps) = \left( h_{y} (\eps) \right)_{y \in \T}$ for the expected absorption times given in \eqref{eqn:hy} has a Laurent series expansion with the order of the pole being one. 
    Moreover, for $j \in \J$ and $y \in \tilde{C}_j$,
    \begin{equation}
    \label{eqn:mainthmETCj}
        h_y (\eps) = \hat{h}_j \eps^{-1} + O(1),
    \end{equation}
    and for each $y \in \tilde{\T}$,
    \begin{equation}
    \label{eqn:mainthmET}
        h_y (\eps) = \left( \sum_{j \in \J} (\tilde{L}_j)_y \hat{h}_j \right) \eps^{-1} + O(1).
    \end{equation}
    On the other hand, if $\abs{\J} = 0$, then for each $y \in \T$, $h_y (\eps)$ has a power series expansion and
    \begin{equation*}
        h_y (\eps) = \E_y^0 \left[ \tau^0 \right] + O(\eps).
    \end{equation*}
\end{theorem}

Theorem \ref{thm:MFPT} says that when $\abs{\J} \geq 1$ and $\eps > 0$ is small, the expected value of the first time that $X^\eps$ exits from the transient set $\T$ starting from $y \in \tilde{\CC}_j$ is approximately $1/\eps$ times the expected value of the first time for $\hat{X}_\mathcal{R}$ to exit the set $\J$ starting from $j$, and the expected value of the first time that $X^\eps$ exits from the transient set $\T$ starting from $y \in \tilde{\T}$ is approximately $1/\eps$ times the weighted sum of the expected values of the first time for $\hat{X}_\mathcal{R}$ to exit the set $\J$ starting from each state $j \in \J$, where the weights are given by the probability for $X^0$ to be absorbed in $\tilde{\CC}_j$ starting from $y$.

\begin{theorem}
\label{thm:CMFPT}
Suppose Assumption \ref{ass:QRstructure} holds. Consider $i \in \I$ and $0 < \eps < \eps_0$. Suppose $\abs{\J} \geq 1$. 
Then, for $j \in \J$ and $y \in \tilde{C}_j$, if $j$ is reachable from $i$ in $\hat{X}_\mathcal{R}$ (i.e., $\hat{\mathcal{L}}_{ji} > 0$), then
\begin{equation}
\label{eqn:mainthmCETCj}
    g_{yi} (\eps) = \hat{g}_{ji} \eps^{-1} + O(1),
\end{equation}
and for each $y \in \tilde{\T}$ such that $(L_{i})_y + \sum_{j \in \J} (\tilde{L}_j)_y \hat{\mathcal{L}}_{ji} > 0$ (i.e., $\mathcal{L}_{yi} (\eps)$ is order one),
\begin{equation}
\label{eqn:mainthmCET}
    g_{yi} (\eps) = \left( \frac{\sum_{j \in \J} (\tilde{L}_j)_y (\hat{\phi}_{i})_j}{(L_{i})_y + \sum_{j \in \J} (\tilde{L}_j)_y ( \hat{\mathcal{L}}_{i} )_j } \right) \eps^{-1} + O(1),
\end{equation}
where $\hat{\mathcal{L}}_{i} = \left( \hat{\mathcal{L}}_{ji} \right)_{j \in \J}$ and $\hat{\phi}_i = - \left( \Mo_0 T_1 \Wo_0 \right)^{-1} \hat{\mathcal{L}}_{i}$.
On the other hand, if $\abs{\J} = 0$, then for each $y \in \T$ such that $(L_i)_y =\PP_y^0 [X^0(\tau^0) \in \CC_i]  > 0$, 
\begin{equation*}
    g_{yi} (\eps) = \E_y^0 \left[ \tau^0 | X^0 (\tau^0) \in \CC_i \right] + O(\eps).
\end{equation*}
\end{theorem}

The conditional expected value of the first time that $X^\eps$ exits from the transient set $\T$ given that $X^\eps$ is absorbed in $\CC_i$ starting from $y \in \T$ is defined only when there is a positive probability that $X^\eps$ is absorbed in $\CC_i$, i.e., when $\mathcal{L}_{yi} (\eps) > 0$. In Theorem \ref{thm:CMFPT}, we further assume $\mathcal{L}_{yi} (\eps)$ is order one. Theorem \ref{thm:CMFPT} says that when $\abs{\J} \geq 1$ and $\eps > 0$ is small, the conditional expected value of the first time that $X^\eps$ exits from the transient set $\T$ given that $X^\eps$ is absorbed in $\CC_i$ starting from $y \in \tilde{\CC}_j$ is approximately $1/\eps$ times the conditional expected value of the first time for $\hat{X}_\mathcal{R}$ to exit the set $\J$ given that $\hat{X}_\mathcal{R}$ is absorbed in state $i$ starting from $j$. The conditional expected value of the first time that $X^\eps$ exits from the transient set $\T$ given that $X^\eps$ is absorbed in $\CC_i$ starting from $y \in \tilde{\T}$ is given by \eqref{eqn:mainthmCET} and, generally speaking, is not approximately $1/\eps$ times the weighted sum of the conditional expected values of the first exit time in $\hat{X}_\mathcal{R}$. In our evolutionary model for fragmented population,  for $y \in \tilde{\T}$ that we are interested in looking at, it is true in the special cases that $g_{yi} (\eps)$ is $1/\eps$ times the conditional expected value of the first time for $\hat{X}_\mathcal{R}$ to exit the set $\J$ given that $\hat{X}_\mathcal{R}$ is absorbed in state $i$ starting from $j$.

\begin{remark}
    Fix $j \in \J$. For $y \in \tilde{\CC}_j$, the zeroth order term of \eqref{eqn:mainthmAPCj} for $\mathcal{L}_{yi} (\eps)$ does not depend on $y$. Further suppose $\abs{\J} \geq 1$. Then, for $y \in \tilde{\CC}_j$, the leading term of \eqref{eqn:mainthmETCj} for $h_y (\eps)$ does not depend on $y$, and for $y \in \tilde{\CC}_j$ where $j$ is such that $\hat{\mathcal{L}}_{ji} > 0$, the leading term of \eqref{eqn:mainthmCETCj} for $g_y (\eps)$ does not depend on $y$.
\end{remark}

Proofs of Theorems \ref{thm:AP}--\ref{thm:CMFPT} are given in \cref{sec:proofs}.
In the following section \ref{sec:general},  we apply Theorems \ref{thm:AP}--\ref{thm:CMFPT} to study the general evolutionary models \eqref{eqn:DemeReactionsMain}--\eqref{eqn:ReactionRates}. In particular, this involves finding the probabilities for mutant fixation and the (unconditional and conditional) expected time for fixation. A more explicit exposition for deriving these results for the two-deme model introduced in section \ref{sec:2D} is provided in the SI.

\section{Wild-type and Mutant Dynamics for $D$ Demes}
\label{sec:general}

We now apply our theorems to analyze the general evolutionary model for fragmented population introduced in section \ref{sec:model}. 
 Here, we focus on the case where the migration network is strongly connected, i.e., each deme is
reachable from any other demes through a sequence of migration events with positive rates.
Recall that the family of continuous time Markov chains $\{ X^\eps: \eps \geq 0 \}$ has a common finite state space
\begin{equation*}
    \X = \{(w_1,m_1,w_2,m_2,\dots,w_D,m_D): 1 \leq w_\ell+m_\ell \leq K_\ell \text{ for each } \ell \in \{ 1,2,\dots,D\} \}.
\end{equation*}
Consider a state where only one type of species is present in each deme. We shall designate the deme configuration of this state with a set that consists of the index numbers of the mutant demes for this configuration. For example, we designate the configuration $wwwww$ with $\emptyset$, the configuration $mmmmm$ with $\{ 1,2,3,4,5 \}$ and the configuration $wwwmm$ with $\{ 4,5 \}$. Then, for $\mathcal{M} \subseteq \{ 1,2,\cdots, D\}$, the collection of all states that have the same configuration $\mathcal{M}$ is
\begin{eqnarray*}
    &&\X_\mathcal{M} = \left\{ x=(w_1,m_1,w_2,m_2,\dots,w_D,m_D) \in \X :  \right.  \\
    &&  \left. \qquad\qquad\qquad\text{if } \ell \in \mathcal{M}, w_\ell=0 \text{ and } 1 \leq m_\ell \leq  K_\ell, \text{and if } \ell \notin \mathcal{M}, 1 \leq w_\ell \leq  K_\ell \text{ and } m_\ell = 0 \right\}.
\end{eqnarray*}

When $\eps > 0$, the chain $X^\eps$ has two recurrent classes $\CC_{\emptyset} = \X_{\emptyset}$ and $\CC_{\{ 1,2,\dots,D\}}= \X_{\{ 1,2,\dots,D\}}$ where
\begin{equation*}
\begin{aligned}
\X_{\emptyset} &= \{(w_1,0,w_2,0,\dots,w_D,0): 1 \leq w_1 \leq K_1, 1 \leq w_2 \leq K_2, \dots, 1 \leq w_D \leq K_D \}, \\
\X_{\{ 1,2,\dots,D\}} &= \{(0,m_1,0,m_2,\dots,0,m_D): 1 \leq m_1 \leq K_1, 1 \leq m_2 \leq K_2, \dots, 1 \leq m_D \leq K_D \},   
\end{aligned}
\end{equation*}
and the set of transient states for $X^\eps$ is $\T = \X \setminus \left( \X_{\emptyset} \cup \X_{\{ 1,2,\dots,D\}} \right)$. 
When $\eps = 0$, $\CC_{\emptyset}$ and $\CC_{\{ 1,2,\dots,D\}}$ are still two recurrent classes for $X^0$. In addition, there are $2^D-2$ more recurrent classes for $X^0$, which are $\tilde{\CC}_{\mathcal{M}} =\X_{\mathcal{M}}: \mathcal{M} \subsetneq \{ 1,2,\dots,D\}$ and $\mathcal{M} \neq \emptyset$. The set of transient states for $X^0$ is $\tilde{\T} = \X \setminus \left( \cup_{\mathcal{M} \subseteq \{ 1,2,\dots,D\}} \X_{\mathcal{M}} \right)$.

In the absence of migration events (i.e., when $\eps = 0$), each deme evolves independently, and so the dynamics in deme $\ell$ can be described by the system presented in section \ref{sec:1D} with $K=K_\ell$.
To describe the quantities introduced in section \ref{sec:1D} for deme $\ell$, we add a superscript $\ell$ to the notations. More precisely, we use $\qss{w}{\ell} (\cdot)$ and $\qss{m}{\ell} (\cdot)$ to denote the corresponding stationary distributions for the wild-type and for the mutant in deme $\ell$, respectively, and we use $\fpone{m}{\ell} (\cdot,\cdot)$ to denote the probabilities of mutant fixation in deme $\ell$. Similarly, for $x = (n,k) \in \X$, we use $\fpone{w}{\ell} (n,k)$ to denote the probability of wild-type fixation starting from $n$ wild-type individuals and $k$ mutant individuals in deme $\ell$.
For $\ell \in \{ 1,2,\dots,D \}$,
we let 
\begin{equation}
\label{eqn:srho1}
    \size{w}{\ell} = \sum_{n=2}^{K_\ell} n \qss{w}{\ell} (n) \qquad \text{ and } \qquad \fp{m}{\ell} = \sum_{n=1}^{K_\ell-1} \qss{w}{\ell} (n) \fpone{m}{\ell} (n,1).
\end{equation}
This $\size{w}{\ell}$ is approximately the expected number of wild-type individuals in deme $\ell$ under its stationary distribution, provided $\qss{w}{\ell} (1)$ is small, and $\fp{m}{\ell}$ is the probability of mutant fixation when a single mutant is introduced from outside deme $\ell$ to the wild-type population in deme $\ell$ under its stationary distribution.
Similarly, for $\ell \in \{ 1,2, \dots, D \}$, we let 
\begin{equation}
\label{eqn:srho2}
    \size{m}{\ell} = \sum_{n=2}^{K_\ell} n \qss{m}{\ell} (n) \qquad \text{ and } \qquad \fp{w}{\ell} = \sum_{n=1}^{K_\ell-1} \qss{m}{\ell} (n) \fpone{w}{\ell} (1,n).
\end{equation}
Since each deme evolves independently when $\eps=0$, the stationary distribution $\tilde{\Pn}_{\mathcal{M}}$ supported on $\tilde{\CC}_{\mathcal{M}}$ for $X^0$ is such that for $x=(w_1,m_1,w_2,m_2,\dots,w_D,m_D) \in \tilde{\CC}_{\mathcal{M}}$, 
\begin{equation}
\label{eqn:nuM}
\tilde{\Pn}_{\mathcal{M}} (x) =  \left( \prod_{i \in \mathcal{M}} \qss{m}{i} (m_i) \right) \left( \prod_{j \notin \mathcal{M}} \qss{w}{j} (w_j) \right),
\end{equation}
and moreover, for $y=x + \basisvec{m}{k}$ where $x \in \tilde{\CC}_{\mathcal{M}}$ $k \notin \mathcal{M}$ and $\basisvec{m}{k}$ is the standard unit basis vector in $\R^{2D}$ for the coordinate designated for the mutant in deme $k$, we have
\begin{equation}
\label{eqn:PM}
    \PP_y^0 [X^0(\tau^0) \in \tilde{\CC}_{\mathcal{M} \cup \{ k \}}] = \fpone{m}{k} (w_k,1).
\end{equation}

Using the construction in section \ref{sec:reducedMC}, we obtain a reduced process $\hat{X}_\mathcal{R}$ that is a continuous time Markov chain with state space $\mathcal{R} = \{ \mathcal{M} : \mathcal{M} \subseteq \{ 1,2,\cdots, D\} \}$, which is the power set of $\{ 1,2,\cdots, D\}$. The states $\emptyset$ and $\{ 1,2,\cdots, D\}$ are absorbing and the other states in $\hat{X}_\mathcal{R}$ are transient. For states $\mathcal{M}_1,\mathcal{M}_2 \in \mathcal{R}$, the infinitesimal transition rate from $\mathcal{M}_1$ to $\mathcal{M}_2$ in $\hat{X}_\mathcal{R}$ is positive only if $\abs{\mathcal{M}_1 \Delta \mathcal{M}_2} = 1$\footnote{The symmetric difference of two sets $\mathcal{M}_1$ and $\mathcal{M}_2$ is defined by $\mathcal{M}_1 \Delta \mathcal{M}_2 = \left( \mathcal{M}_1 \setminus \mathcal{M}_2 \right) \cup \left( \mathcal{M}_2 \setminus \mathcal{M}_1 \right)$}.
For $\ell \in \mathcal{M}$ and $k \in \{ 1,2,\cdots, D\} \setminus \mathcal{M}$, let 
\begin{equation*}
    v^{(\ell,k)}_m = \basisvec{m}{k} - \basisvec{m}{\ell} \in \Z^{2n} \quad \text{ and } \quad v^{(k,\ell)}_w = \basisvec{w}{\ell} - \basisvec{w}{k} \in \Z^{2n}
\end{equation*}
be the reaction direction for a mutant migrating from deme $\ell$ to deme $k$ and the reaction direction for a wild-type migrating from deme $k$ to deme $\ell$. Here, $\basisvec{m}{\ell}$, $\basisvec{m}{k}$, $\basisvec{w}{\ell}$, $\basisvec{w}{k}$ are the standard unit basis vectors in $\R^{2D}$ for the coordinates designated for the mutant in deme $\ell$, the mutant in deme $k$, the wild-type in deme $\ell$, the wild-type in deme $k$, respectively.
Then, from \eqref{eqn:MW}--\eqref{eqn:QR} and \eqref{eqn:nuM}--\eqref{eqn:PM}, we have that  for $\mathcal{M} \subsetneq \{ 1,2,\cdots, D\}$ and $k \in \{ 1,2,\cdots, D\} \setminus \mathcal{M}$,
\small
\begin{eqnarray}
    && \left(Q_\mathcal{R} \right)_{\mathcal{M},\mathcal{M} \cup \{ k \}} \nonumber \\
    &=& \sum_{\ell \in \mathcal{M}} \sum_{x \in \X_\mathcal{M}} \tilde{\Pn}_{\mathcal{M}} (x) Q^{(1)}_{x,x+v^{(\ell,k)}_m} \PP_{x+v^{(\ell,k)}_m}^0 \left[ X^0(\tau^0) \in \X_{\mathcal{M} \cup \{ k \}} \right] \nonumber \\
    &=& \sum_{\ell \in \mathcal{M}} \sum_{(w_1,m_1,\dots,w_D,m_D) \in \X_\mathcal{M}} \left( \left( \prod_{i \in \mathcal{M}} \qss{m}{i} (m_i) \right) \left( \prod_{j \notin \mathcal{M}} \qss{w}{j} (w_j) \right) \mi{m}{\ell}{k} m_\ell \one_{\{ m_\ell > 1 , w_k < K_k \}} \fpone{m}{k} (w_k,1) \right) \nonumber
\end{eqnarray}
\normalsize
and so 
\begin{equation}
\label{eqn:QRk}
    \left(Q_\mathcal{R} \right)_{\mathcal{M},\mathcal{M} \cup \{ k \}} = \sum_{\ell \in \mathcal{M}}  \mi{m}{\ell}{k} \size{m}{\ell} \fp{m}{k}.
\end{equation}
This means that the infinitesimal transition rate in $\hat{X}_\mathcal{R}$ by which deme $k$ changes from being all-wild-type  to all-mutant, is obtained by summing, over all mutant demes, of the product the migration rate from that deme,  the (approximate) expected size of the mutant deme, and the probability that a migrant mutant ultimately fixates in a wild-type deme $k$.
Similarly, for $\mathcal{M} \subseteq \{ 1,2,\cdots, D\}$ such that $\mathcal{M} \neq \emptyset$ and $\ell \in \mathcal{M}$, we have
\begin{equation}
\label{eqn:QRl}
    \left(Q_\mathcal{R} \right)_{\mathcal{M},\mathcal{M} \setminus \{\ell\}} = \sum_{k \notin \mathcal{M}} \mi{w}{k}{\ell} \size{w}{k} \fp{w}{\ell}. 
\end{equation}

\begin{remark}
As a special case, if we further assume that for all $\ell,k \in \{ 1,2,\dots,D \}$, 
\begin{equation*}
    K_\ell=K, \quad \rr{w}{\ell} = \bar{r}_w, \quad \rr{m}{\ell} = \bar{r}_m, \quad \dd{w}{\ell} = \bar{d}_w, \quad \dd{m}{\ell} = \bar{d}_m, \quad \mi{w}{\ell}{k} = \bar{\mu}_m, \quad \mi{m}{\ell}{k} = \bar{\mu}_m,
\end{equation*}
where $K,\bar{r}_w,\bar{r}_m,\bar{d}_w,\bar{d}_m,\bar{\mu}_w,\bar{\mu}_m$ are positive constants that do not depend on deme indices, we have that $\size{w}{\ell} = \mathfrak{s}_{w}$, $\size{m}{\ell} = \mathfrak{s}_{m}$, $\fp{w}{k} = \rho_w$, $\fp{m}{k} = \rho_m$ for some positive constants $\mathfrak{s}_{w}$, $\mathfrak{s}_{m}$, $\rho_w$, $\rho_m$ that do not depend on deme indices.
Then, \eqref{eqn:QRk}--\eqref{eqn:QRl} become
\begin{equation*}
    \left(Q_\mathcal{R} \right)_{\mathcal{M},\mathcal{M} \cup \{ k \}} = \bar{\mu}_m \mathfrak{s}_{m} \rho_m \abs{\mathcal{M}} \qquad \text{ and } \qquad \left(Q_\mathcal{R} \right)_{\mathcal{M},\mathcal{M} \setminus \{\ell\}} = \bar{\mu}_w \mathfrak{s}_{w} \rho_w (D - \abs{\mathcal{M}}).
\end{equation*}
In this case, we may project $\hat{X}_\mathcal{R}$ down to a one-dimensional Markov chain, which counts the number of mutant demes. This projection is a birth-death process whose state space is $\{ \abs{\mathcal{M}} : \mathcal{M} \in \mathcal{R} \} = \{ 0,1,\dots,D \}$. At state $n$, the birth rate is $\bar{\mu}_m \mathfrak{s}_{m} \rho_m (D - n) n$ and death rate is $\bar{\mu}_w \mathfrak{s}_{w} \rho_w (D - n) n$.
\end{remark}

\subsection{Mutation Fixation Analysis}
We first analyze mutation fixation starting from a deterministic state $y$ with only one mutant in deme $\ell$, i.e., $y=x + \basisvec{m}{\ell}$ for $x \in \X_{\emptyset}$ and $\basisvec{m}{\ell}$ being the standard unit basis vector in $\R^{2D}$ for the coordinate designated for the mutant in deme $\ell$.
We observe that for $\mathcal{M} \subseteq \{ 1,2,\dots,D \}$ such that $\mathcal{M} \neq \emptyset$ and $\mathcal{M} \neq \{ \ell \}$, we have $\PP_y^0 [X^0(\tau^0) \in \X_{\mathcal{M}}] = 0$. Moreover, since each deme evolves independently when $\eps=0$, $\PP_y^0 [X^0(\tau^0) \in \tilde{\CC}_{\{ \ell \}}] = \fpone{m}{\ell} (w_\ell,1)$.  Since the migration network is strongly connected, we have from \eqref{eqn:QRk}--\eqref{eqn:QRl} that $\emptyset$ and $\{ 1,2,\cdots, D\}$ are the only two absorbing states in $\hat{X}_\mathcal{R}$ and the other states are transient, and so Assumption \ref{ass:QRstructure} holds. Thus, from \eqref{eqn:mainthmAP} in Theorem \ref{thm:AP}, we have 
\begin{equation}
\label{eqn:LhatN}
        \mathcal{L}_{y,\{1,2,\dots,D\}} (\eps) = \fpone{m}{\ell} (w_\ell,1) \hat{\mathcal{L}}_{\{\ell\},\{1,2,\dots,D\}} + O(\eps),
\end{equation}
which represents a two-stage decomposition of the mutant fixation probability: starting from a single mutant in deme $\ell$, the probability of mutant fixation for the whole system is, to the leading order, the product of the probability of its fixation in the founding deme and  the fixation probability in the reduced Markov chain (i.e., the probability of spread from one mutant deme to all demes). Here, $\hat{\mathcal{L}}_{\{\ell\},\{1,2,\dots,D\}} > 0$. Further, from \eqref{eqn:mainthmET} in Theorem \ref{thm:MFPT}, 
\begin{equation}
\label{eqn:hhatN}
    h_y (\eps) = \fpone{m}{\ell} (w_\ell,1) \hat{h}_{\{\ell\}} \eps^{-1} + O(1),
\end{equation}
which means that the leading contribution to the mean absorption time comes from trajectories that first establish a mutant deme (with probability $\fpone{m}{\ell} (w_\ell,1)$) and then  evolve on the slow deme-level timescale (of the order $\eps^{-1}$). Finally, 
from \eqref{eqn:mainthmCET} in Theorem \ref{thm:CMFPT}, 
\begin{equation}
\label{eqn:ghatN}
    g_{y,\{1,2,\dots,D\}} (\eps) = \frac{\hat{\phi}_{\{\ell\},\{1,2,\dots,D\}}}{\hat{\mathcal{L}}_{\{\ell\},\{1,2,\dots,D\}}} \eps^{-1} + O(1) = \hat{g}_{\{\ell\},\{1,2,\dots,D\}} \eps^{-1} + O(1),
\end{equation}
where the second equality holds from Lemma \ref{lem:MTWInvertible}. Like equation (\ref{eqn:LhatN}), this is also a two-stage result, which means that conditional on success, the time to global fixation is determined by the second (global) stage, namely the spread of mutant demes through the metapopulation.

We now apply the general results of equations  (\ref{eqn:LhatN}-\ref{eqn:ghatN}) to the biologically relevant scenario, in which a mutant arises through immigration and is introduced into a randomly chosen deme. Assume the wild-type population has come to equilibrium, i.e., $X^\eps$ is restricted to the recurrent class $\CC_{\emptyset}$. For $\eps \geq 0$, $X^\eps$ restricted to $\CC_{\emptyset}$ is positive recurrent and has a unique stationary distribution $\Pn^\eps_{\emptyset}$. Thus, by Theorem 4.1 in Campos et al. \cite{bib:epifinite}, we have for $x = (w_1,0,w_2,0,\dots,w_D,0) \in \X_{\emptyset}$,
\begin{equation}
\label{eqn:WQSSDD}
    \Pn^\eps_{\emptyset} (x) = \Pn_{\emptyset} (x) + O(\eps)= \prod_{k=1}^D \qss{w}{k} (w_k) + O(\eps),
\end{equation}
where $\Pn_{\emptyset}$ is the stationary distribution supported on $\CC_\emptyset$ for $X^0$.
We then consider an immigration event where a mutant is introduced from an external source and the probability it falls in a particular deme is proportional to the carrying capacity of that deme. From \eqref{eqn:LhatN} and \eqref{eqn:WQSSDD}, we have that the probability that both demes are fixated by the mutant is
\begin{eqnarray}
    && \sum_{x \in \X_\emptyset} \Pn^\eps_{\emptyset} (x) \sum_{\ell=1}^{D} \frac{K_\ell}{K_1+K_2+\dots+K_D} \mathcal{L}_{x+\basisvec{m}{\ell},\{1,2,\dots,D\}} (\eps) \one_{\{x+\basisvec{m}{\ell} \in \X\}} \nonumber \\
    &=& \sum_{\ell=1}^{D} \frac{K_\ell}{K_1+K_2+\dots+K_D} \fp{m}{\ell} \hat{\mathcal{L}}_{\{\ell\},\{1,2,\dots,D\}} + O(\eps). \label{eqn:FPfinal}
\end{eqnarray}
From \eqref{eqn:hhatN} and \eqref{eqn:WQSSDD}, we have that the expected time until either the mutant or the wild-type fixates is
\begin{eqnarray}
    && \sum_{x \in \X_\emptyset} \Pn^\eps_{\emptyset} (x) \sum_{\ell=1}^{D} \frac{K_\ell}{K_1+K_2+\dots+K_D} h_{x+\basisvec{m}{\ell}} (\eps) \one_{\{x+\basisvec{m}{\ell} \in \X\}} \nonumber \\
    &=& \left( \sum_{\ell=1}^{D} \frac{K_\ell}{K_1+K_2+\dots+K_D} \fp{m}{\ell} \hat{h}_{\{\ell\}} \right) \eps^{-1} + O(1). \label{eqn:EFTfinal}
\end{eqnarray}
From \eqref{eqn:ghatN} and \eqref{eqn:WQSSDD}, we have that the conditional expected time for mutant fixation is
\begin{eqnarray}
	&& \frac{\sum_{x \in \X_\emptyset} \Pn^\eps_{\emptyset} (x) \sum_{\ell=1}^{D} \frac{K_\ell}{K_1+K_2+\dots+K_D} g_{x+\basisvec{m}{\ell},\{1,2,\dots,D\}} (\eps) \mathcal{L}_{x+\basisvec{m}{\ell},\{1,2,\dots,D\}} (\eps) \one_{\{x+\basisvec{m}{\ell} \in \X\}}}{\sum_{x \in \X_\emptyset} \Pn^\eps_{\emptyset} (x) \sum_{\ell=1}^{D} \frac{K_\ell}{K_1+K_2+\dots+K_D} \mathcal{L}_{x+\basisvec{m}{\ell},\{1,2,\dots,D\}} (\eps) \one_{\{x+\basisvec{m}{\ell} \in \X\}}} \nonumber \\
	\qquad &=& \left( \sum_{\ell=1}^{D}  \frac{\mathcal{K}_\ell}{\mathcal{K}_1+\mathcal{K}_2+\dots+\mathcal{K}_D} \hat{g}_{\{\ell\},\{1,2,\dots,D\}} \right) \eps^{-1} + O(1), \label{eqn:CEFTfinal}
\end{eqnarray}
where $\mathcal{K}_\ell = K_\ell \fp{m}{\ell} \hat{\mathcal{L}}_{\{\ell\},\{1,2,\dots,D\}}$ for $\ell = 1,2,\dots,D$.

\subsection{Examples - dynamics for three demes with different migration patterns}
In this section, we look at two examples of our evolutionary model for fragmented population in three demes with a total carrying capacity $K_1+K_2+K_3 =100$.

\noindent \textbf{Example 1:} We first consider the quasi-neutral case (see \cite{bib:2007PQ_b,bib:2010PQP, bib:2025KW}), where 
\begin{equation*}
    \rr{m}{\ell} = \tau \rr{w}{\ell}, \qquad \dd{m}{\ell} = \tau \dd{w}{\ell}, \qquad \mi{m}{\ell}{k} = \tau \mi{w}{\ell}{k}
\end{equation*}
for some constant $\tau>0$, i.e., the mutant basic division and death rates are proportionally scaled with respect to those of the wild-type. We study the mutant fixation when $K_1,K_2,K_3$ vary. In particular, we look at three different migration patterns: (1) $\mi{w}{\ell}{k} = 1$ for all $\ell \neq k$, (2) $\mi{w}{1}{3} = \mi{w}{3}{2} = \mi{w}{2}{1} = 1$ and $\mi{w}{1}{2} = \mi{w}{2}{3} = \mi{w}{3}{1} = 0$, (3) $\mi{w}{1}{2}  = \mi{w}{2}{1} = \mi{w}{2}{3} = \mi{w}{3}{2} = 1$ and $\mi{w}{1}{3}  = \mi{w}{3}{1} = 0$. The migration patterns and the corresponding reduced Markov chain structures are shown in Figure \ref{fig:XhatI-3D}. The heatmaps for the fixation probabilities and the (unconditional and conditional) expected fixation times are shown in Figures \ref{fig:3DFP}--\ref{fig:3DCEFT}.

\vspace{-0.3cm}
\begin{figure}[H]
    \centering    \includegraphics[width=\linewidth]{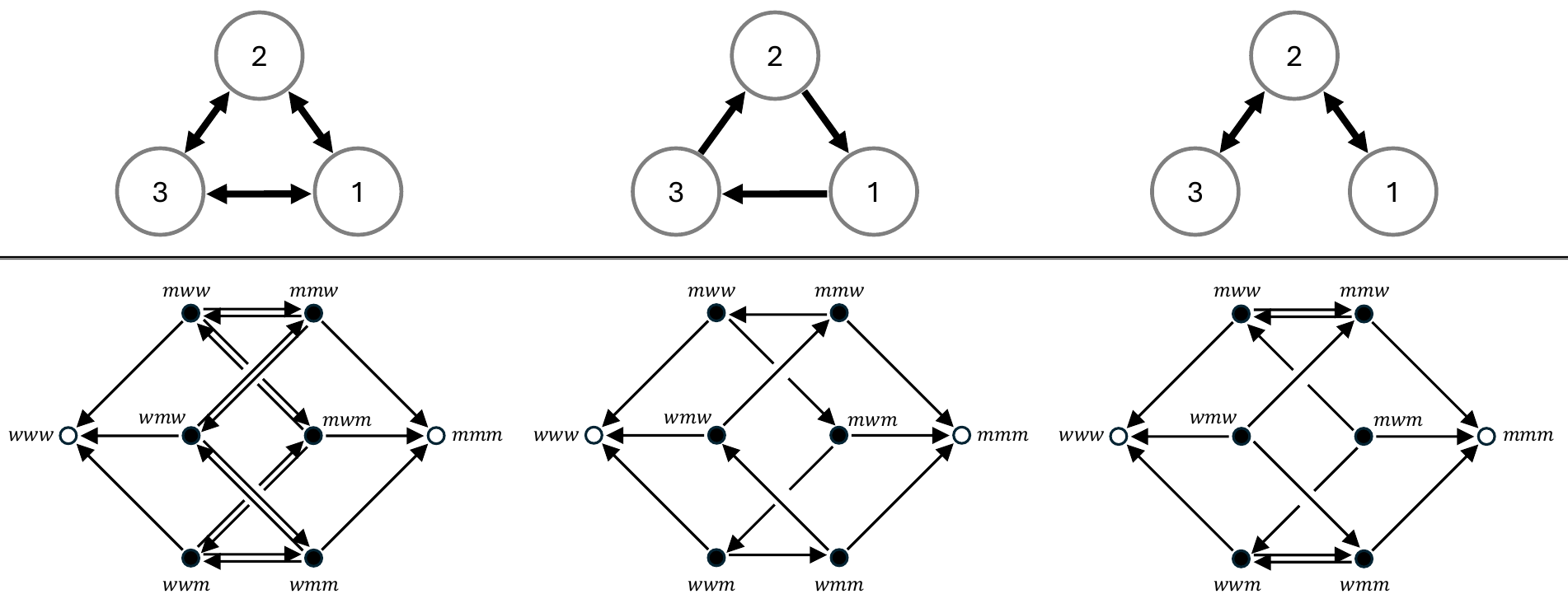}
    \vspace{-0.6cm}
    \caption{\textbf{Markov chain graphs of $\hat{X}_\mathcal{R}$ for three-deme population with different migration patterns.} Three types of migration patterns are considered: (left) $\mi{w}{\ell}{k} = 1$ for all $\ell \neq k$, (middle) $\mi{w}{1}{3} = \mi{w}{3}{2} = \mi{w}{2}{1} = 1$, (right) $\mi{w}{1}{2}  = \mi{w}{2}{1} = \mi{w}{2}{3} = \mi{w}{3}{2} = 1$. The infinitesimal generator for $\hat{X}_\mathcal{R}$ is given by \eqref{eqn:QRk}--\eqref{eqn:QRl}, and the non-zero infinitesimal transitions are shown in the corresponding lower panel for each type of migration pattern.} 
    \label{fig:XhatI-3D}
\end{figure}
\vspace{-0.6cm}
\begin{figure}[H]
    \centering
    \includegraphics[width=0.8\linewidth]{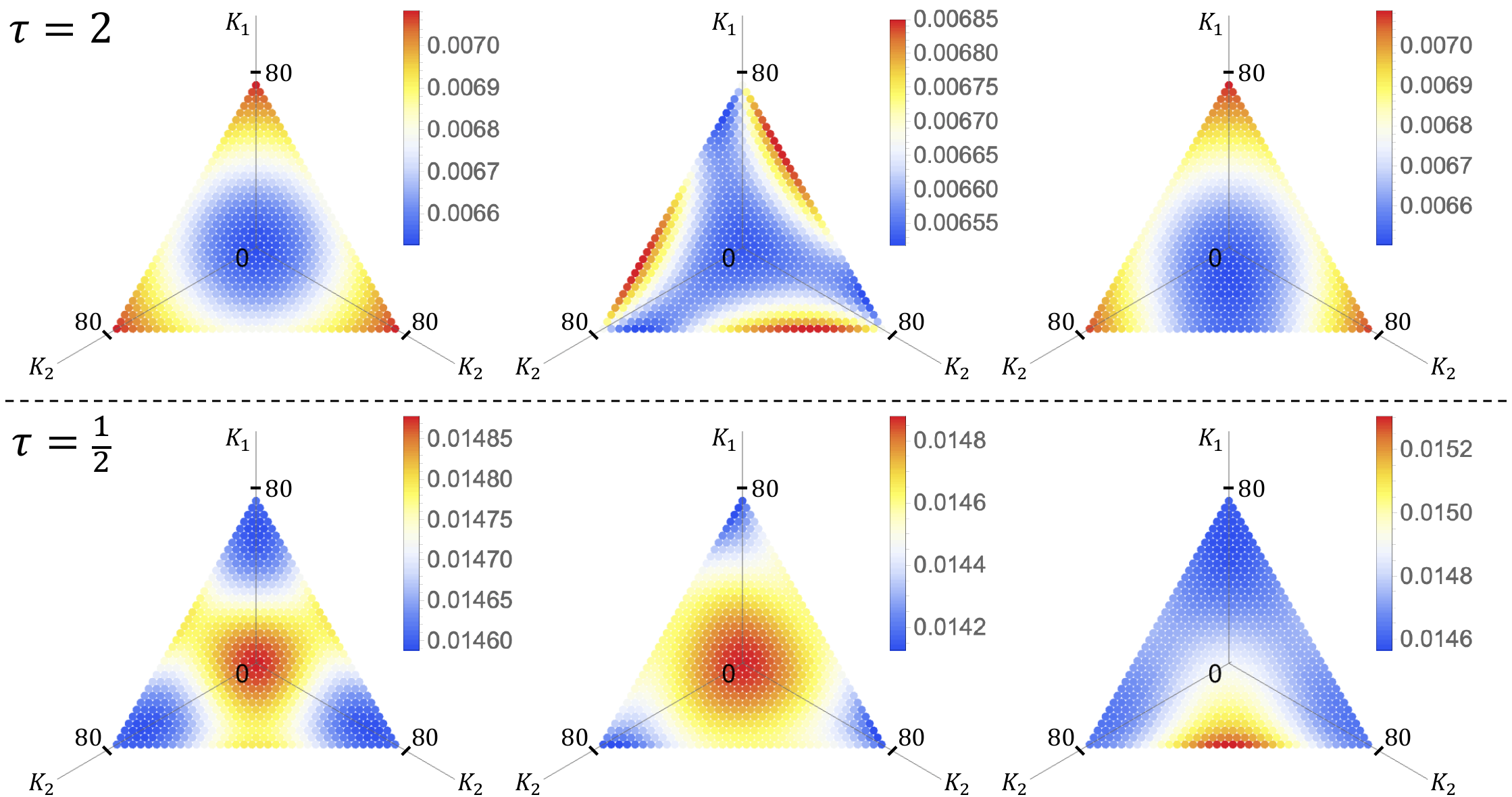}
    \vspace{-0.1cm}
    \caption{\textbf{Probabilities of mutant fixation for three-deme population with a total capacity $100$.} Here, $K_1+K_2+K_3 = 100$ and we plot for even integers $K_1, K_2, K_3 \geq 10$. For $\ell \in \{ 1,2,3\}$, $\rr{w}{\ell}=10$, $\rr{m}{\ell}= 10 \tau$, $\dd{w}{\ell}=1$ and $\dd{m}{\ell} = \tau$. For $\ell \neq k$, $\mi{m}{\ell}{k} = \tau \mi{w}{\ell}{k}$. Three types of migration patterns are considered: (left) $\mi{w}{\ell}{k} = 1$ for all $\ell \neq k$, (middle) $\mi{w}{1}{3} = \mi{w}{3}{2} = \mi{w}{2}{1} = 1$, (right) $\mi{w}{1}{2}  = \mi{w}{2}{1} = \mi{w}{2}{3} = \mi{w}{3}{2} = 1$. In each heatmap, we plot the value of the leading term in \eqref{eqn:FPfinal}, which is approximately the probability of mutant fixation when $\eps$ is small. In SI - section \ref{sec:Fig3}, these heatmaps are presented with a different color scale, where the fixation probabilities are compared to that for the well-mixed system.} 
    \label{fig:3DFP}
\end{figure}

\begin{figure}[H]
    \centering
    \vspace{-0.5cm}
    \includegraphics[width=0.8\linewidth]{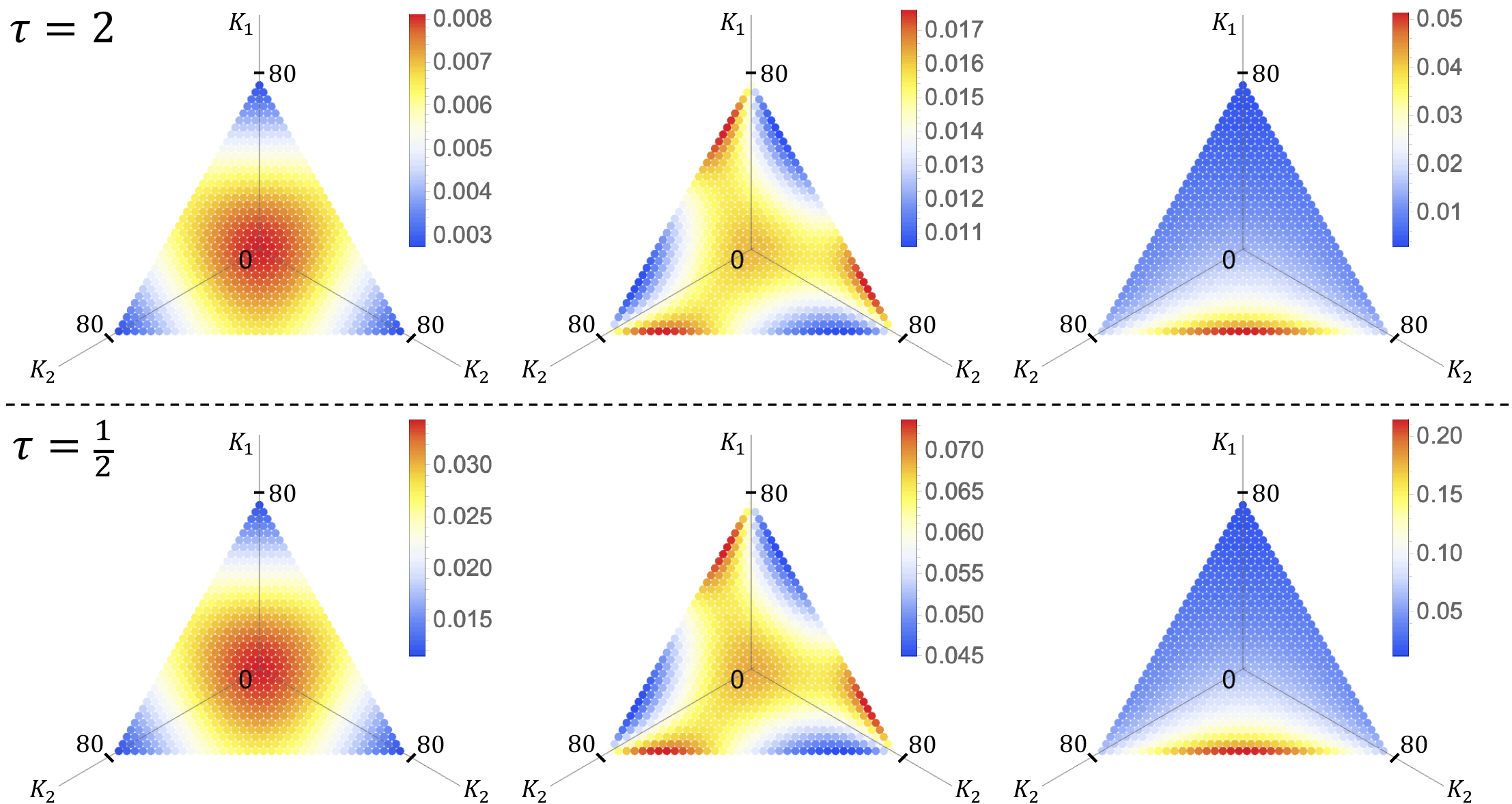}
    
    \caption{\textbf{Leading coefficients of expected fixation times (when either the wild-type or the mutant fixates) for three-deme population with a total capacity $100$.} 
    For parameter values, see Figure \ref{fig:3DFP}. In each heatmap, we plot the value of the leading coefficient in \eqref{eqn:EFTfinal}, which times $1 / \eps$ is approximately the expected time until either the wild-type or the mutant fixates for the system \eqref{eqn:DemeReactionsMain}--\eqref{eqn:ReactionRates} when $\eps$ is small. Please note the scale difference; the expected absorption times for the decelerated mutant ($\tau=1/2$) are approximately $4$ times larger than that for the corresponding accelerated mutant ($\tau=2$), see SI - section \ref{sec:Fig45} for a rationale for this near proportional scaling.}
    \label{fig:3DEFT}
\end{figure}
\vspace{-1cm}
\begin{figure}[H]
    \centering
    \includegraphics[width=0.88\linewidth]{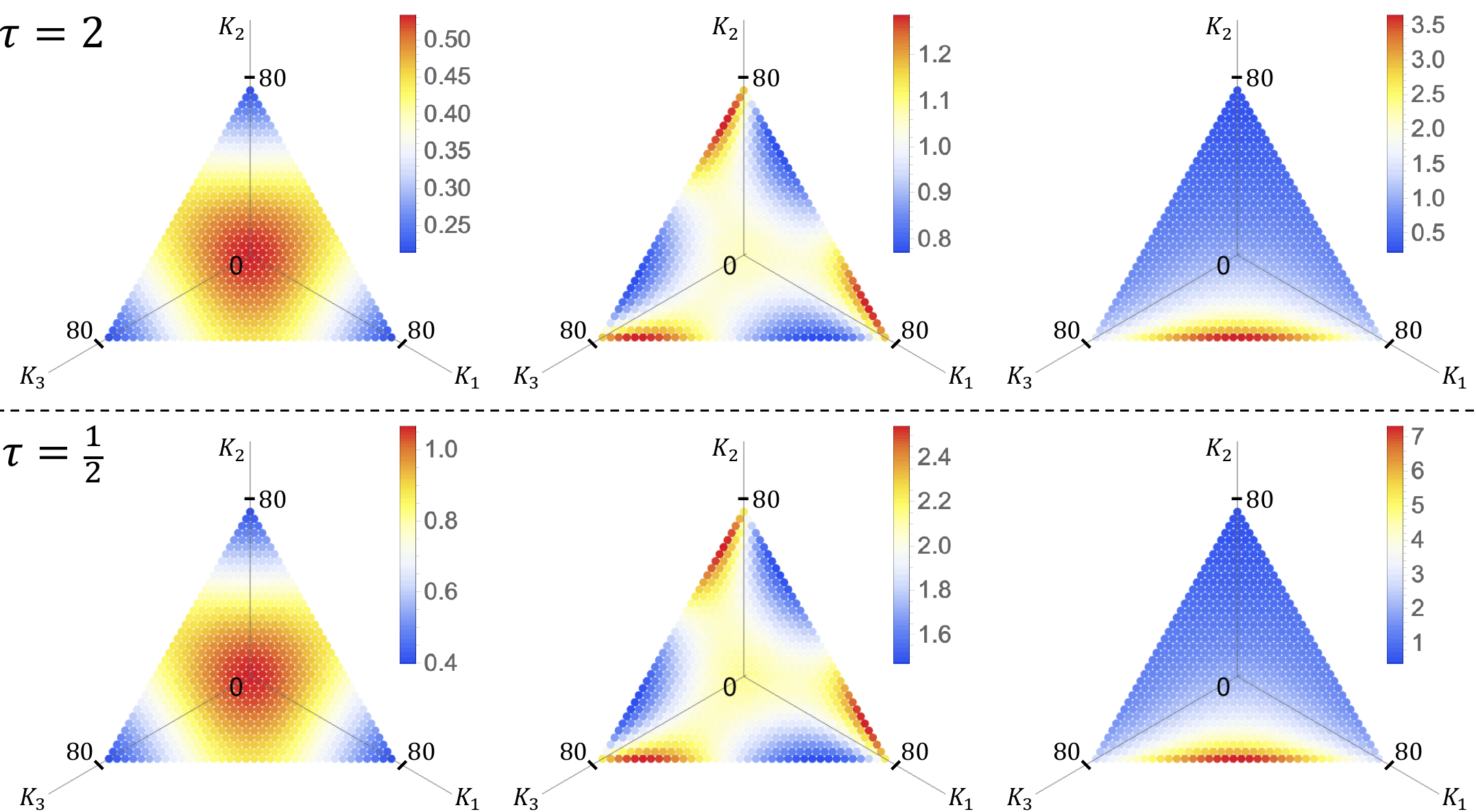}
    \vspace{-0.1cm}
    \caption{\textbf{Leading coefficients of conditional expected times for mutant fixation for three-deme population with a total capacity $100$.} For parameter values, see Figure \ref{fig:3DFP}.
    In each heatmap, we plot the value of the leading coefficient in \eqref{eqn:CEFTfinal}, which times $1 / \eps$ is approximately the conditional expected time for fixation given that the mutant fixates when $\eps$ is small. Please note the scale difference; the expected conditional fixation times for the decelerated mutant ($\tau=1/2$) are approximately $2$ times larger than that for the corresponding accelerated mutant ($\tau=2$), see SI - section \ref{sec:Fig45} for a rationale for this near proportional scaling.
    }
    \label{fig:3DCEFT}
\end{figure}

We observe that under constant total carrying capacity, both the size distribution of individual demes and specific migration patterns affect the probability and timing of mutant fixation. For example, for the first (most symmetric) migration pattern, for mutants with proportionally accelerated rates (Figure \ref{fig:3DFP}, $\tau=2$), the probability of fixation tends to decrease for equal deme sizes, while  it tends to increase for equal deme sizes if the mutants are decelerated ($\tau=1/2$). On the contrary, the deme-size dependence of absorption and conditional fixation times for accelerated and decelerated mutants is similar (comparing the top and the bottom rows of Figures  \ref{fig:3DEFT}--\ref{fig:3DCEFT}). In fact, as shown in SI - section \ref{sec:Fig45}, the mean absorption time for quasi-neutral  mutants characterized by factors $\tau_0>1$ and $1/\tau_0$ differ by a multiplicative factor of $\tau_0^2$, and the mean conditional fixation times differ by a multiplicative factor of $\tau_0$. Both temporal characteristics vary significantly with deme carrying capacities and migration patterns.

\medskip
\noindent \textbf{Example 2:} We consider a case where the mutant has a growth disadvantage, i.e.,
\begin{equation*}
    \rr{m}{\ell} = (1+s) \rr{w}{\ell}, \qquad \dd{m}{\ell} = \dd{w}{\ell}, \qquad \mi{m}{\ell}{k} = a \mi{w}{\ell}{k}
\end{equation*}
with $s<0$ and $a>0$. Here, we assume that the rate at which any individual migrate depends on whether the individual is a wild-type or a mutant, but is independent of where the individual is, i.e., $\mi{w}{\ell}{k} = 1$ for all $\ell \neq k$. The heatmaps for the fixation probabilities with various values for $a$ are shown in Figure \ref{fig:FPmig}.

\begin{figure}[H]
    \centering
    \includegraphics[width=\linewidth]{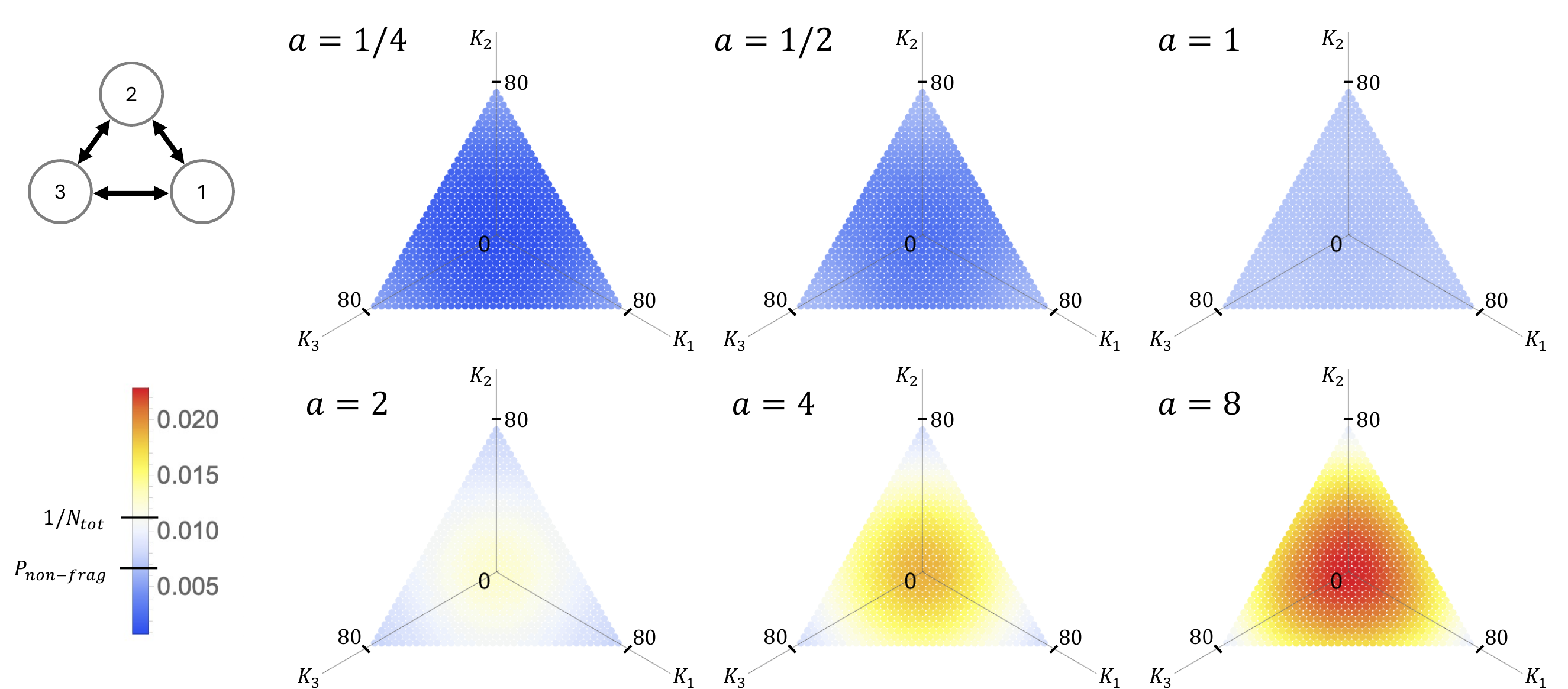}
    \caption{\textbf{Probabilities of mutant fixation for three-deme population with a total capacity $100$ where the mutant has a growth disadvantage.}
    Here, $K_1+K_2+K_3 = 100$ and we plot for even integers $K_1, K_2, K_3 \geq 10$. For $\ell \in \{ 1,2,3\}$, $\rr{w}{\ell}=10$, $\rr{m}{\ell}= 9.9$, $\dd{w}{\ell}= \dd{m}{\ell} = 1$. For $\ell \neq k$, $\mi{w}{\ell}{k} = 1$ and $\mi{m}{\ell}{k} = a$.  In each heatmap, we plot the value of the leading term in \eqref{eqn:FPfinal}, which is approximately the probability of mutant fixation when $\eps$ is small. We compare these fixation probabilities with $1/N_{tot}$ and $P_{\text{non-frag}}$, where $N_{tot} \vcentcolon= K_{tot} \left(1 -\frac{d_w}{r_w} \right)$ is approximately the mean population size for the wild-type and $P_{\text{non-frag}}$ is the fixation probability of a well-mixed population with a carrying capacity of $100$ starting from a single mutant.
    }
    \label{fig:FPmig}
\end{figure}

We observe that for $a \le 1$, the mutants act as disadvantageous, with fixation probability equal or smaller than that calculated for a non-fragmented system. But as $\mu_m$ increases beyond $\mu_w$ ($a>1$), we observe an increase in the fixation probability, which, for sufficiently large $a$, exceeds $1/N_{tot}$ where $N_{tot} \vcentcolon= K_{tot} (1 -d_w/r_w)$ is approximately the mean population size for the wild-type. In other words, an ability to migrate faster than the wild-type makes these disadvantageous mutants behave as if they were advantageous. Interestingly, this transition first happens in the middle of the simplex that characterizes the deme sizes, meaning that the effect of fast migration is the strongest in a system of equal demes.  These results add to the theory of evolution of dispersal and competition–colonization trade-offs (see e.g. \cite{fronhofer2024evolutionary} for a recent review). The behavior in the middle of the simplex shown in Figure \ref{fig:FPmig} corresponds to the work of  \cite{marrec2026impact}, and regions away from the middle demonstrate the effect of imbalanced deme sizes.

\section{Discussion}

In this paper we studied a logistic birth-death process (with density-dependent divisions) in a fragmented population with migration, in the absence of deme extinction. We showed (Theorems \ref{thm:AP}, \ref{thm:MFPT}, and \ref{thm:CMFPT}) that when migration is much slower than the within-deme dynamics ($\eps \to 0$), the original stochastic process on the full state space can be replaced by a reduced Markov chain whose states correspond to  configurations of homogeneous (fully wild-type or fully mutant) demes. Fixation probabilities and absorption times of the full process converge to those of this reduced process. Equations (\ref{eqn:LhatN})--(\ref{eqn:ghatN}) reveal a two-stage structure of the invasion process in the rare-migration limit. A newly introduced mutant must first establish within its founding deme, after which the subsequent dynamics are governed by a reduced Markov chain. We further obtained leading-order approximations for the mutant fixation probability (equation (\ref{eqn:FPfinal})), mean absorption time (equation (\ref{eqn:EFTfinal})), and mean conditional fixation time (equation (\ref{eqn:CEFTfinal})), where we assumed that the mutant was initiated by an external migration event to a deme chosen with the probability proportional to its carrying capacity.

Results derived here are quite general and hold for demes of unequal carrying capacities, under any (deme-dependent) division and death rates of wild-types and mutants, and any migration rates as long as the migration network is strongly connected. 
Moreover, the theory applies  to the general version of the reaction network model \eqref{eqn:DemeReactionsMain}, not just the specific birth-death process mentioned.

Because population fragmentation is ubiquitous in nature, the theoretical framework developed here has potential applications across a wide range of biological and biomedical problems. Such systems include cancer cell populations evolving in spatially structured tissues, microbial communities and bacterial colonies, ecological populations inhabiting fragmented landscapes, and species whose dispersal patterns have been modified by anthropogenic habitat fragmentation and other forms of human intervention.

This work considers the parameter regime where migration events happen at slower rates than reproduction and death events, in which case $\eps$ is small and the proposed power series expansion is valid. When migration events happen at much faster rates than reproduction and death events, it is possible to use the singular perturbation theory proposed in this paper to study the system by considering a different $Q(\eps)$. On the other hand, if the time-scale of migration events is comparable with that of reproduction and death events, this framework is no longer applicable. The other limitation is deme extinction. In our model, we assume that the population does not die out in each deme. The assumption is no longer methodologically sound if the deme size is small and/or the death events happen at a comparable rate to the reproduction events ($\frac{d}{r} \approx 1$).

\appendix

\section{Additional Lemmas and Proofs}
\label{sec:proofs}
We define $T(\eps) = T_0 + \eps T_1$ and for $i \in \I$, define $R_i (\eps) = R_{0,i} + \eps R_{1,i}$. 
Recall that for each $i \in \I$, $\one_i$ is the column vector of all ones with size $\abs{\CC_i}$  and $L_i = \left(\PP_y^0 [X^0(\tau^0) \in \CC_i] \right)_{y \in \tilde{\T}}$. Let \begin{equation*}
\mathcal{L}_{0,i}  
= 
\begin{tikzpicture}[baseline={-0.5ex},mymatrixenv]
\matrix [mymatrix,inner sep=5pt] (m)  
{
0 \\
\vdots \\
0 \\
L_i \\
};
\draw[densely dashed] (-0.4,-0.6) -- (0.4,-0.6);
\mymatrixbraceleft{1}{1}{$\tilde{\CC}_1$}
\mymatrixbraceleftempty{2}{2}{$\vdots$}
\mymatrixbraceleft{3}{3}{$\tilde{\CC}_{\numI}$}
\mymatrixbraceleft{4}{4}{$\tilde{\T}$}
\end{tikzpicture} \in \R^{\abs{\T}}.
\end{equation*}
Then, it can be checked that for each $i \in \I$, $J_i \in \R^{\abs{\J} \times 1}$ defined in \eqref{eqn:Ji} satisfies 
\begin{equation*}
    J_i = \Mo_0 \left( R_{1,i} \one_i + T_1 \mathcal{L}_{0,i} \right).
\end{equation*}
We will use these functions and expressions throughout this section.
To prove Theorems \ref{thm:AP}--\ref{thm:CMFPT}, one can invoke the results stated in Lemmas \ref{lem:APandMFPT} and \ref{lem:TepsInverse}. In particular, Lemma \ref{lem:TepsInverse}
characterizes the leading terms for $T^{-1}(\eps)$, when plugged into the expressions \eqref{eqn:lemmaLi}--\eqref{eqn:lemmagi} in Lemma \ref{lem:APandMFPT}, one would be able to get the results stated in Theorems \ref{thm:AP}--\ref{thm:CMFPT} after simplifying the expressions. The details of the algebraic simplification are provided in the SI. Lemmas \ref{lem:AP0}--\ref{lem:MWbases} provide technical results used in the proof of Lemma \ref{lem:TepsInverse}. We start the section with Lemma \ref{lem:QR} that shows the well-posedness of the reduced Markov chain $\hat{X}_\mathcal{R}$.

\begin{lemma}
\label{lem:QR}
    The matrix $Q_{\mathcal{R}}$ given in \eqref{eqn:QR} is an infinitesimal generator.
\end{lemma}

\begin{proof}
    Let $\one_\J$, $\one_{\tilde{\T}}$ and  $\one_{\T}$ be the vectors of all ones with sizes $|\J|$, $|\tilde{\T}|$ and $|\T|$, respectively. 
    To see the row sums of $Q_{\mathcal{R}}$ are zeros, we observe that
    \begin{eqnarray*}
        \Mo_0 T_1 \Wo_0 \one_\J + \sum_{i \in \I} J_i &=& \Mo_0 T_1 \Wo_0 \one_\J + \sum_{i \in \I} \Mo_0 \left( R_{1,i} \one_i + T_1 \mathcal{L}_{0,i} \right) \\
        &=& \Mo_0 \left( \sum_{i \in \I} R_{1,i} \one_i + T_1 \left( \Wo_0 \one_\J + \sum_{i \in \I} \mathcal{L}_{0,i} \right) \right)
    \end{eqnarray*}
    \begin{eqnarray*}
        \quad \qquad &=& \Mo_0 \left( \sum_{i \in \I} R_{1,i} \one_i + T_1 \one_\T \right) = 0,
    \end{eqnarray*}
    where we rearrange terms to get the second equality, the third equality holds since for each $y \in \tilde{\T}$,
    \begin{equation*}
        1 = \PP_y^0 [X^0(\tau^0) \in \X \setminus \tilde{\T} ] = \left( \sum_{i \in \I} \PP_y^0 [X^0(\tau^0) \in \CC_i] \right) + \left( \sum_{j \in \J} \PP_y^0 [X^0(\tau^0) \in \tilde{\CC}_j] \right)
    \end{equation*}
    and so $\left( \sum_{i \in \I} L_i \right) + \left( \sum_{j \in \J} \tilde{L}_j \right) = \one_{\tilde{\T}}$ which implies $ \left( \sum_{i \in \I} \mathcal{L}_{0,i} \right) + \Wo_0 \one_\J = \one_\T$, and the last equality holds since the row sums of $Q^{(1)}$ being zeros implies that $\left( \sum_{i \in \I} R_{1,i} \one_i \right) + T_1 \one_\T = 0$.
    
    For $j \in \J$ and $i \in \I$, since
    \begin{equation*}
        0 \leq Q_{xy} (\eps) = \eps \left( R_{ji}^{(1)} \right)_{xy} \qquad \text{ for } x \in \tilde{\CC}_j, y \in \CC_i, \eps \in (0,\eps_0),
    \end{equation*}
    \begin{equation*}
        0 \leq Q_{xy} (\eps) = \eps \left( \tilde{U}_{j}^{(1)} \right)_{xy} \qquad \text{ for } x \in \tilde{\CC}_j, y \in \tilde{\T}, \eps \in (0,\eps_0),
    \end{equation*}
    we have $R_{ji}^{(1)} \geq 0$ and $\tilde{U}_{j}^{(1)} \geq 0$, and since also $\tilde{\Pn}_j \geq 0$, $\one_i \geq 0$ and $L_i \geq 0$, we have
    \begin{equation*}
        \left( Q_{\mathcal{R}} \right)_{ji} = \tilde{\Pn}_j \left( R_{ji}^{(1)} \one_i + \tilde{U}_{j}^{(1)} L_i \right) \geq 0.
    \end{equation*}
    For $j \in \J$ and $k \in \J \setminus \{ j \}$, since
    \begin{equation*}
        0 \leq Q_{xy} (\eps) = \eps \left( \tilde{E}_{jk}^{(1)} \right)_{xy} \qquad \text{ for } x \in \tilde{\CC}_j, y \in \tilde{\CC}_k, \eps \in (0,\eps_0),
    \end{equation*}
    we have $\tilde{E}_{jk}^{(1)} \geq 0$, and since also $\tilde{\Pn}_j \geq 0$, $\tilde{\one}_k \geq 0$, $\tilde{U}_{j}^{(1)} \geq 0$ and $\tilde{L}_k \geq 0$, we have
    \begin{equation*}
        \left( Q_{\mathcal{R}} \right)_{jk} = \tilde{\Pn}_j \left( \tilde{E}_{jk}^{(1)} \tilde{\one}_k + \tilde{U}_{j}^{(1)} \tilde{L}_k \right) \geq 0.
    \end{equation*}
    Thus, $Q_{\mathcal{R}}$ is an infinitesimal generator.
\end{proof}

\begin{lemma}
\label{lem:APandMFPT}
    Consider $0 < \eps < \eps_0$. 
    The matrix $T(\eps)$ is invertible. Moreover, for $i \in \I$, $\mathcal{L}_{i} (\eps) = \left( \mathcal{L}_{yi} (\eps) \right)_{y \in \T}$ given in \eqref{eqn:Lyi} is 
    \begin{equation}
    \label{eqn:lemmaLi}
        \mathcal{L}_{i} (\eps) = - T^{-1} (\eps) R_i (\eps) \one_i,
    \end{equation}
    and $h (\eps) = \left( h_{y} (\eps) \right)_{y \in \T}$ given in \eqref{eqn:hy} is 
    \begin{equation}
    \label{eqn:lemmah}
        h (\eps) = - T^{-1} (\eps) \one_\T,
    \end{equation}
    where $\one_\T$ is the vector of all ones with size $|\T|$.
    For $i \in \I$ and $y \in \T$, if $\mathcal{L}_{yi} (\eps) >0$, then $g_{yi} (\eps)$ given in \eqref{eqn:gy} is
    \begin{equation}
    \label{eqn:lemmagi}
        g_{yi} (\eps) = \frac{\phi_{yi} (\eps)}{\mathcal{L}_{yi} (\eps)},
    \end{equation}
    where $\phi_{yi} (\eps) = \left( - T^{-1} (\eps) \mathcal{L}_{i} (\eps) \right)_y$ and $\mathcal{L}_{yi} (\eps) = \left( \mathcal{L}_{i} (\eps) \right)_y = \left( - T^{-1} (\eps) R_i (\eps) \one_i \right)_y$.
\end{lemma}

\begin{proof}
    When $\eps > 0$, $\T$ is the set of all transient states in $X^\eps$. By Lemma S.4 in Bruno et al. \cite{bib:epifinite}, we have that $T(\eps)$ is invertible. The formulae for $\mathcal{L}_{i} (\eps)$ and $h (\eps)$ can be derived through first step analysis (see a proof for this, for example, in Lemma S.4 and Equation (3.2) in Bruno et al. \cite{bib:epifinite}). 
    The formula for $g_{yi} (\eps)$ can be obtained through first step analysis for $\phi_{yi} (\eps) = \E_y \left[ \tau^\eps \one_{ \{ X^\eps (\tau^\eps) \in \CC_i \} } \right]$, where one gets $\phi_i (\eps) = \left(\phi_{yi} (\eps) \right)_{y \in \T} = - T^{-1} (\eps) \mathcal{L}_{i} (\eps)$, along with the fact that $g_{yi} (\eps) = \E_y^\eps \left[ \tau^\eps | X^\eps (\tau^\eps) \in \CC_i \right] = \frac{\E_y^\eps \left[ \tau^\eps \one_{ \{ X^\eps (\tau^\eps) \in \CC_i \} } \right]} {\PP_y^\eps \left[  X^\eps (\tau^\eps) \in \CC_i \right]} = \frac{\phi_{yi} (\eps)}{\mathcal{L}_{yi} (\eps)}$.
\end{proof}

\begin{lemma}
\label{lem:AP0}
    The matrix $\tilde{T}^{(0)}$ is invertible. Moreover, for $i \in \I$ and $j \in \J$,
    \begin{equation*}
        L_i = - \left(\tilde{T}^{(0)} \right)^{-1} R_i^{(0)} \one_i \qquad \text{ and } \qquad \tilde{L}_j = - \left(\tilde{T}^{(0)} \right)^{-1} \tilde{R}_j^{(0)} \tilde{\one}_j.
    \end{equation*}
\end{lemma}

\begin{proof}
    The proof is similar to that of Lemma \ref{lem:APandMFPT}.
\end{proof}

\begin{lemma}
\label{lem:MTWInvertible}
    Suppose Assumption \ref{ass:QRstructure} holds. Then, $\Mo_0 T_1 \Wo_0$ is invertible. 
    Moreover, for $i \in \I$, $\hat{\mathcal{L}}_{i} = \left( \hat{\mathcal{L}}_{ji} \right)_{j \in \J} = - \left( \Mo_0 T_1 \Wo_0 \right)^{-1} J_i$, 
    and $\hat{h} = \left( \hat{h}_j \right)_{j \in \J} = - \left(\Mo_0 T_1 \Wo_0 \right)^{-1} \one_\J$ 
    where $\one_\J$ is the vector of all ones with size $|\J|$.
    For $j \in \J$ and $i \in \I$, if $\hat{\mathcal{L}}_{ji} > 0$, then $\hat{g}_{ji} =\frac{\hat{\phi}_{ji}}{\hat{\mathcal{L}}_{ji}}$
    where $\hat{\phi}_{ji} = \left( - \left( \Mo_0 T_1 \Wo_0 \right)^{-1} \hat{\mathcal{L}}_{i} \right)_j$ and $\hat{\mathcal{L}}_{ji} = \left( - \left( \Mo_0 T_1 \Wo_0 \right)^{-1} J_i \right)_j$.
\end{lemma}

\begin{proof}
    The proof is similar to that of Lemma \ref{lem:APandMFPT}.
\end{proof}

\begin{lemma}
\label{lem:MWbases}
    The null space of $T_0$ has dimension $\numJ$. Moreover, $\Mo_0 T_0 = 0$ and $T_0 \Wo_0 = 0$.
\end{lemma}

\begin{proof}
    For $j \in \J$, since $\tilde{\CC}_j$ is a recurrent class in $X^0$, the dimension of the null space of $\tilde{E}_{jj}^{(0)}$ is $1$. Since also $\tilde{T}^{(0)}$ is invertible by Lemma \ref{lem:AP0}, the dimension of the null space of $T_0$ is $\abs{\J} = \numJ$.
    For $j \in \J$, since $\tilde{\Pn}_j$ is the stationary distribution for the recurrent class $\tilde{\CC}_j$ in $X^0$, we have $\tilde{\Pn}_j \tilde{E}_{jj}^{(0)} = 0$, and so $\Mo_0 T_0 = 0$.
    For $j \in \J$, since $\tilde{E}_{jj}^{(0)} \tilde{\one}_{j} = 0$ and by Lemma \ref{lem:AP0}, $\tilde{R}_j^{(0)} \tilde{\one}_{j} + \tilde{T}^{(0)} \tilde{L}_{j} = \tilde{R}_j^{(0)} \tilde{\one}_{j} + \tilde{T}^{(0)} \left(\tilde{T}^{(0)} \right)^{-1} \tilde{R}_j^{(0)} \tilde{\one}_j = 0$, we have $T_0 \Wo_0 = 0$.
\end{proof}

\begin{lemma}
\label{lem:TepsInverse}
    Suppose Assumption \ref{ass:QRstructure} holds. There exists $\eps_1 \in (0,\eps_0)$ such that
    \begin{equation}
        \label{eqn:TLaurentSeries}
        T^{-1}(\eps) = \sum_{k=-1}^{\infty} \eps^k B^{(k)}, \quad  0 < \eps < \eps_1,
    \end{equation}
    where $\{B^{(k)}:\: k \geq -1\}$ is a sequence of matrices in $\R^{|\T| \times |\T|}$ and $\sum_{k = -1}^{\infty} \eps^k\norm{B^{(k)}} < \infty$ for each $0 < \eps < \eps_1$. Furthermore,
    \begin{equation}
    \label{eqn:B-1Formula}
        B^{(-1)} = \Wo_0 \left(\Mo_0 T_1 \Wo_0 \right)^{-1} \Mo_0
    \end{equation}
    and 
    \begin{equation}
    \label{eqn:B0Formula}
        B^{(0)} = - \left( I - B^{(-1)} T_1 \right) \left( \left( \Wo_0 \Mo_0 -T_0 \right)^{-1} - \Wo_0 \Mo_0 \right)
        \left( I - T_1 B^{(-1)}\right).
    \end{equation}
\end{lemma}

\begin{remark}
    In Lemma \ref{lem:TepsInverse}, we slightly abuse the notation when $\abs{\J} = 0$, in which case $\Wo_0$ and $\Mo_0$ are degenerate, $T_0$ is invertible, and \eqref{eqn:B-1Formula}--\eqref{eqn:B0Formula} become $B^{(-1)} = 0$ and $B^{(0)} = (T_0)^{-1}$, which means $T^{-1}(\eps) = (T_0)^{-1} + O(\eps)$.
\end{remark}

\begin{proof}
    Since $T(\eps) = T_0 + \eps T_1$ is invertible for each $0 < \eps < \eps_0$ (see Lemma \ref{lem:APandMFPT}), we have by Theorem 2.4 in Avrachenkov el al. \cite{Avrachenkov2013} and Proposition 1.1 in Bruno et al. \cite{bib:epifinite} that
    there exists $\eps_1 \in (0,\eps_0)$ and $p \in \Z_+$ such that $T^{-1}(\eps)$ has a Laurent series expansion 
    \begin{equation}
        \label{eqn:TLaurentSeriesP}
        T^{-1}(\eps) = \sum_{k=-p}^{\infty} \eps^k B^{(k)}, \quad  0 < \eps < \eps_1.
    \end{equation}

    We first consider the case where $\J$ is non-empty. By Lemma \ref{lem:MTWInvertible}, we have that $\Mo_0 T_1 \Wo_0$ is invertible under Assumption \ref{ass:QRstructure}. We are going to show that $\Mo_0 T_1 \Wo_0$ being invertible implies that $p=1$, which is stated in Theorem 2.9 in Avrachenkov el al. \cite{Avrachenkov2013}. We further give explicit expression for the first two leading terms in \eqref{eqn:TLaurentSeries} where we utilize the deviation matrix for $X^0$ to serve the purpose of the generalized inverse in \cite{Avrachenkov2013}.

    For this, since $T^{-1}(\eps) T(\eps) = I$, we have by \eqref{eqn:TLaurentSeriesP} that if $p=0$, then $B^{(-p)} T_0 = I$.
    Since the dimension of the null space of $T_0$ is $|\J| \geq 1$ (see Lemma \ref{lem:MWbases}), there is no $B^{(-p)} \in \R^{|\T| \times |\T|}$ such that $B^{(-p)} T_0 = I$. Thus, $p>0$ and
    \begin{equation*}
        B^{(-p)} T_0 = 0.
    \end{equation*}
    Since rows of $\Mo_0$ form a basis for the null space of the transpose of $T_0$ (see Lemma \ref{lem:MWbases}), we have $B^{(-p)} = C^{(-p)} \Mo_0$ for some non-zero $C^{(-p)} \in \R^{|\T| \times |\J|}$. Since $T_0 \Wo_0 = 0$ (see Lemma \ref{lem:MWbases}), we have
    \begin{equation}
    \label{eqn:Cp}
        \left( B^{(-p+1)} T_0 + B^{(-p)} T_1 \right) \Wo_0 = \left( B^{(-p+1)} T_0 + C^{(-p)} \Mo_0 T_1 \right) \Wo_0 = 0 + C^{(-p)} \Mo_0 T_1 \Wo_0.
    \end{equation}
    Since $\Mo_0 T_1 \Wo_0$ is invertible (see Lemma \ref{lem:MTWInvertible}) and $C^{(-p)} \neq 0$, \eqref{eqn:Cp} cannot be zero. Thus, $\left(B^{(-p+1)} T_0 + B^{(-p)} T_1 \right) \eps^{-p+1}$ is the leading term of the series expansion for $T^{-1}(\eps) T(\eps)$, 
    which implies the order of the pole is one, i.e., $p=1$, and 
    \begin{equation}
    \label{eqn:Bp}
        B^{(0)} T_0 + B^{(-1)} T_1 = I.
    \end{equation}
    Then, by \eqref{eqn:Cp}--\eqref{eqn:Bp}, we have that $C^{(-1)} = \Wo_0 \left(\Mo_0 T_1 \Wo_0 \right)^{-1}$ and so
    \begin{equation}
    \label{eqn:B-1}
        B^{(-1)} = \Wo_0 \left(\Mo_0 T_1 \Wo_0 \right)^{-1} \Mo_0.
    \end{equation}
    Since rows of $\Mo_0$ form a basis for the null space of the transpose of $T_0$ (see Lemma \ref{lem:MWbases}), we have by \eqref{eqn:Bp} that
    \begin{eqnarray}
    \label{eqn:C0}
        B^{(0)} = C^{(0)} \Mo_0 + \left( I - B^{(-1)} T_1 \right) (T_0)^\dagger,
    \end{eqnarray}
    for some $C^{(0)} \in \R^{|\T| \times |\J|}$, where $(T_0)^\dagger$ is a generalized inverse\footnote{A generalized inverse $(T_0)^\dagger$ of $T_0$ is such that $T_0 (T_0)^\dagger T_0 = T_0$.} of $T_0$.
    Since $T_0 \Wo_0 = 0$ and $B^{(1)} T_0 + B^{(0)} T_1 = 0$ (because $T^{-1}(\eps) T(\eps) = I$), we have $C^{(0)}$ in \eqref{eqn:C0} must satisfies
    \begin{equation*}
        0 = (B^{(1)} T_0 + B^{(0)} T_1) \Wo_0 = 0 + C^{(0)} \Mo_0 T_1 \Wo_0 + \left( I - B^{(-1)} T_1 \right) (T_0)^\dagger T_1 \Wo_0.
    \end{equation*}
    Since $\Mo_0 T_1 \Wo_0$ is invertible, we have $C^{(0)} = - \left( I - B^{(-1)} T_1 \right) (T_0)^\dagger T_1 \Wo_0 \left(\Mo_0 T_1 \Wo_0 \right)^{-1}$, and so
    \small
    \begin{equation*}
        B^{(0)} = \left( I - B^{(-1)} T_1 \right) (T_0)^\dagger 
        \left( - T_1  \Wo_0 \left(\Mo_0 T_1 \Wo_0 \right)^{-1} \Mo_0 + I \right) = \left( I - B^{(-1)} T_1 \right) (T_0)^\dagger 
        \left( I - T_1 B^{(-1)} \right).
    \end{equation*}
    \normalsize
    Lastly, to show $B^{(0)}$ is given by \eqref{eqn:B0Formula}, it suffices to show that $\left( - \left( \left( \Wo_0 \Mo_0 -T_0 \right)^{-1} - \Wo_0 \Mo_0 \right) \right)$ is a generalized inverse of $T_0$. For this, observe that 
    for each $j \in \J$, since $\tilde{\Pn}_j$ is the unique stationary distribution for the recurrent class $\tilde{\CC}_j$ in $X^0$, we have 
    $\tilde{\Pn}_j \tilde{E}_{jj}^{(0)} = 0$, which implies $u = 0$ is the only solution to $\left(\tilde{\one}_{j} \tilde{\Pn}_{j} - \tilde{E}_{jj}^{(0)} \right) u = 0$.
    Thus, $\tilde{\one}_{j} \tilde{\Pn}_{j} - \tilde{E}_{jj}^{(0)}$ is invertible. By Lemma \ref{lem:AP0}, $\tilde{T}^{(0)}$ is invertible. 
    Let $\tilde{D}_j = \left(\tilde{\one}_{j} \tilde{\Pn}_{j} - \tilde{E}_{jj}^{(0)} \right)^{-1} - \tilde{\one}_{j} \tilde{\Pn}_{j}$ be the deviation matrix for a continuous time Markov chain with infinitesimal generator $\tilde{E}_{jj}^{(0)}$. One can check that $\tilde{\Pn}_{j} \tilde{D}_j = 0$ and $\tilde{E}_{jj}^{(0)} \tilde{D}_j \tilde{E}_{jj}^{(0)} = -\tilde{E}_{jj}^{(0)}$.
    Then, using the Schur complement to obtain the inverse of the invertible matrix $\Wo_0 \Mo_0 -T_0$, we have
    \begin{eqnarray}
        && \left( \Wo_0 \Mo_0 -T_0 \right)^{-1} - \Wo_0 \Mo_0 \nonumber \\
        &=& 
        \begin{tikzpicture}[baseline={0ex},mymatrixenv]
        \matrix [mymatrix,inner sep=5pt] (m)  
        {
        \tilde{\one}_{1} \tilde{\Pn}_{1} - \tilde{E}_{11}^{(0)} & 0 & 0 & 0 \\
        0 & \ddots & 0 & \vdots \\
        0 & 0 & \tilde{\one}_{\numJ} \tilde{\Pn}_{\numJ} - \tilde{E}_{\numJ\numJ}^{(0)} & 0 \\
        \tilde{L}_{1} \tilde{\Pn}_{1} - \tilde{R}_1^{(0)} & \cdots & \tilde{L}_{\numJ} \tilde{\Pn}_{\numJ} - \tilde{R}_{\numJ}^{(0)} & - \tilde{T}^{(0)} \\
        };
        \draw[densely dashed] (2,-1.3) -- (2,1.3);
        \draw[densely dashed] (-3,-0.65) -- (3,-0.65);
        \end{tikzpicture}^{-1} 
        -
        \begin{tikzpicture}[baseline={0ex},mymatrixenv]
        \matrix [mymatrix,inner sep=5pt] (m)  
        {
        \tilde{\one}_{1} \tilde{\Pn}_{1}  & 0 & 0 & 0 \\
        0 & \ddots & 0 & \vdots \\
        0 & 0 & \tilde{\one}_{\numJ} \tilde{\Pn}_{\numJ} & 0 \\
        \tilde{L}_{1} \tilde{\Pn}_{1} & \cdots & \tilde{L}_{\numJ} \tilde{\Pn}_{\numJ} & 0 \\
        };
        \draw[densely dashed] (1.2,-1.2) -- (1.2,1.2);
        \draw[densely dashed] (-1.9,-0.6) -- (1.9,-0.6);
        \end{tikzpicture} \nonumber \\
        &=& 
        \begin{tikzpicture}[baseline={0ex},mymatrixenv]
        \matrix [mymatrix,inner sep=5pt] (m)  
        {
        \tilde{D}_1 & 0 & 0 & 0 \\
        0 & \ddots & 0 & \vdots \\
        0 & 0 & \tilde{D}_{\numJ} & 0 \\
        \left( - \tilde{T}^{(0)} \right)^{-1} \left(\tilde{R}_1^{(0)} \tilde{D}_1 - L_1 \tilde{\Pn}_{1} \right) & \cdots & \left( - \tilde{T}^{(0)} \right)^{-1} \left(\tilde{R}_{\numJ}^{(0)} \tilde{D}_{\numJ} - L_{\numJ} \tilde{\Pn}_{\numJ} \right) & \left( - \tilde{T}^{(0)} \right)^{-1} \\
        };
        \draw[densely dashed] (4.4,-1.3) -- (4.4,1.3);
        \draw[densely dashed] (-6.4,-0.45) -- (6.4,-0.45);
        \end{tikzpicture} \qquad \qquad \label{eqn:GIforD0}
    \end{eqnarray}
    where we use the fact that $\tilde{L}_j = - \left(\tilde{T}^{(0)} \right)^{-1} \tilde{R}_j^{(0)} \tilde{\one}_j$ for each $j \in \J$ (see Lemma \ref{lem:AP0}). From \eqref{eqn:GIforD0}, one can check that $T_0 \left( - \left( \left( \Wo_0 \Mo_0 -T_0 \right)^{-1} - \Wo_0 \Mo_0 \right) \right) T_0 = T_0$, and thus $B^{(0)}$ can be given by \eqref{eqn:B0Formula}.

    When $|\J| = 0$, $\T$ is the set of all transient states in $X^0$, and so $T_0$ is invertible by Lemma S.4 in Bruno et al. \cite{bib:epifinite}. Thus, $B^{(-p)} T_0 = I$ and so $p=0$, which is consistent with \eqref{eqn:B-1Formula}--\eqref{eqn:B0Formula}.
\end{proof}

\section*{Big Picture Ideas of the Mathematical Methods}
\noindent \textbf{Stochastic model:} In this paper, we analyze evolution in fragmented populations (demes), connected via migration. To monitor the stochastic dynamics in such populations, the numbers of wild-type and mutant individuals located in each deme are tracked over time. This can be modeled as a continuous time Markov chain where a transition in the Markov chain takes place whenever an individual in the population reproduces, dies or migrates. A stochastic model is necessary if one is interested in studying mutant extinction and fixation in structured populations.
	
\noindent \textbf{Model reduction and perturbation analysis:} The model can be rather high dimensional (see for example Figure \ref{fig:XhatI}a), and proper reduction can help understand and analyze the system more efficiently.  Fragmented populations features a separation of timescales between rapid within-deme evolution and slow migration events, governed by a small parameter $\eps$. Intuitively, the dynamics within each deme is likely to equilibrate before a migration event occurs. The averaged behavior within demes can be used to provide coarse-grained approximation of the original detailed dynamics. While coarse-grained approximation has been applied before in different contexts, in this work we provide a rigorous method to perform this model reduction through stochastic singular perturbation analysis. The resulting reduced model is another continuous time Markov chain with a much smaller state space (see for example Figure \ref{fig:XhatI}c) and can be used to provide leading-order asymptotic of the original model. 
	
\noindent \textbf{A few more words on singular perturbation:} It is compelling to study singularly perturbed systems (aside from regularly perturbed systems). To analyze the system $X^\eps$ when $\eps$ is small, one cannot directly look the system $X^0$ when $\eps=0$ since $X^0$ has multiple absorbing communicating classes. On the other hand, the behavior of $X^\eps$ can be explained as a mixture of behaviors of $X^0$ on each absorbing communicating class.
	
\noindent \textbf{Further reading:}
For interested readers, we recommend the following resources:
\begin{enumerate}
	\item for an introduction to population genetics, see \cite{hartl1997principles}, and also \cite{ewens2004mathematical} for some mathematical methods in population genetics;
	\item in this paper, the continuous time Markov chain model is set up in the language of Stochastic Chemical Reaction Networks (SCRNs), and for a general introduction to SCRNs, we refer the reader to \cite{Anderson_Kurtz2015StochAnalysis,Wilkinson2006SMS}; 
	\item singular perturbation was studied in the context of discrete time Markov chains \cite{Avrachenkov2013} and continuous time Markov chains \cite{bib:epifinite}. 
\end{enumerate}

\section*{Acknowledgments}
The authors would like to thank Ruth J. Williams for many helpful discussions and valuable feedback on  this manuscript. We also thank the referee for their valuable comments to improve the manuscript.

\section*{Funding}
Partial support of the NSF grants DMS 2424853 and DMS 2619681 is gratefully acknowledged.

\section*{Author Contributions}
All authors have made substantial intellectual contributions to the study conception, execution, and design of the work. All authors have read and approved the final manuscript.  In addition, the following contributions occurred:  Conceptualization: Yi Fu, Natalia L. Komarova; Methodology: Yi Fu, Natalia L. Komarova; Formal analysis and investigation: Yi Fu; Writing - original draft preparation: Yi Fu; Writing - review and editing: Yi Fu, Natalia L. Komarova;  Supervision: Natalia L. Komarova.

\section*{Access to Code}
The code for generating figures is available at \url{https://github.com/yiiif/MutantFixationInFragmentedPopulations.git}.

\printbibliography[heading=bibintoc, title={References}]
\end{refsection}

\newpage
\begin{refsection}

\section*{Supplementary Information (SI)}

\appendixpageoff
\appendixtitleoff
\renewcommand{\appendixtocname}{Supplementary Information}
\begin{appendices}

\setcounter{section}{19}
\crefalias{section}{supp}
\setcounter{figure}{0}
\renewcommand{\thefigure}{S.\arabic{figure}}

\setcounter{equation}{0}
\renewcommand{\theequation}{S.\arabic{equation}}

\pagenumbering{arabic}
\renewcommand*{\thepage}{\arabic{page}}	

\subsection{Another look at figures}

\subsubsection{Figure \ref{fig:3DFP} for probabilities of mutant fixation}
\label{sec:Fig3}
In Figure \ref{fig:3DFP-SI}, we present the heatmaps shown in Figure \ref{fig:3DFP} but with an alternative color scale. 
\begin{figure}[H]
    \centering
    \includegraphics[width=1\linewidth]{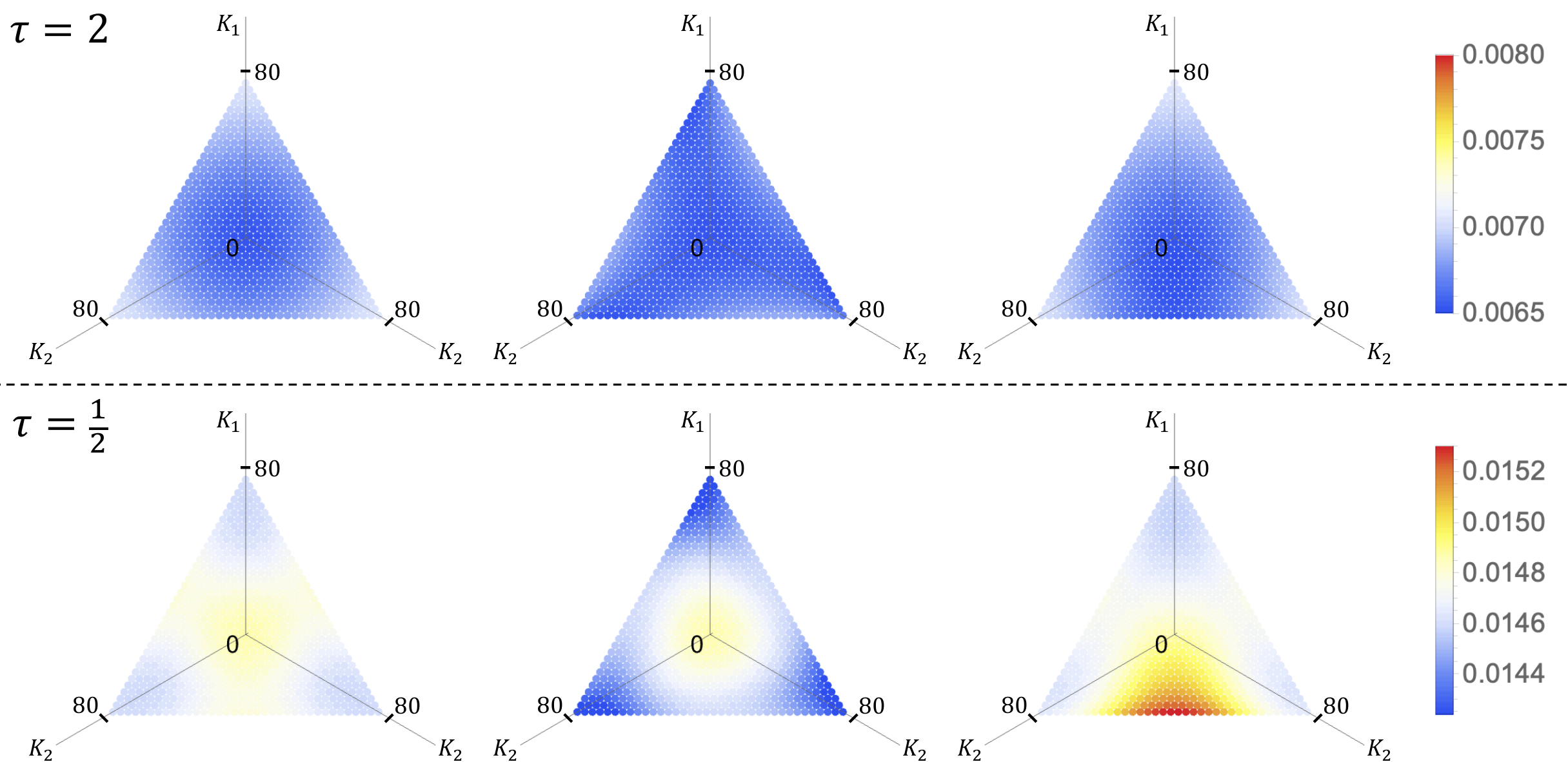}
    \caption{\textbf{Probabilities of mutant fixation for three-deme population with a total capacity $100$.} For parameter values, see Figure \ref{fig:3DFP}. In each heatmap, we plot the value of the leading term in \eqref{eqn:FPfinal}, which is approximately the probability of mutant fixation when $\eps$ is small. In the upper panel, we choose the color scale to encodes the values ranging from $0.00650$ to $0.00803$, where the midpoint ($0.007265$) is the probability of mutant fixation for well-mixed population with $\tau=2$ and carrying capacity being $100$. In the lower panel, we choose the color scale to encodes the values ranging from $0.01424$ to $0.01530$, where the midpoint ($0.01477$) is the probability of mutant fixation for well-mixed population with $\tau=1/2$ and carrying capacity being $100$.} 
    \label{fig:3DFP-SI}
\end{figure}

\subsubsection{Figures \ref{fig:3DEFT}--\ref{fig:3DCEFT} for unconditional and conditional fixation times}
\label{sec:Fig45}

Here, we provide a rationale for the approximately proportional scaling in the fixation times between the $\tau=2$ case and the $\tau=1/2$ case, as seen in Figures \ref{fig:3DEFT}--\ref{fig:3DCEFT}. We assume the quasi-neutral parameter regime as in section \ref{sec:general}:
\begin{equation*}
    \rr{m}{\ell} = \tau \rr{w}{\ell}, \qquad \dd{m}{\ell} = \tau \dd{w}{\ell}, \qquad \mi{m}{\ell}{k} = \tau \mi{w}{\ell}{k}.
\end{equation*}
For $\ell \in \{ 1,2,\dots,D \}$, we let $N_\ell = K_\ell \left( 1-\frac{d_w}{r_w} \right)$. 

When the carrying capacity $K_\ell$ for deme $\ell$ is large, if the deme is occupied by the wild-type, the approximated expected size of the deme is
\begin{equation*}
    \size{w}{\ell} \approx K_\ell \left( 1-\frac{d_w}{r_w} \right) = N_\ell,
\end{equation*}
and the probability that a migrant mutant ultimately fixates in the wild-type deme is
\begin{equation*}
    \fp{m}{\ell} \approx \frac{1}{N_\ell} \frac{2}{\tau+1}.
\end{equation*}
Similarly, if the deme is occupied by the mutant, the approximated expected size of the deme is
\begin{equation*}
     \size{m}{\ell} \approx K_\ell \left( 1-\frac{d_m}{r_m} \right) = K_\ell \left( 1-\frac{\tau d_w}{\tau r_w} \right) = N_\ell,
\end{equation*} 
and the probability that a migrant wild-type ultimately fixates in the mutant deme $\ell$ is
\begin{equation*}
    \fp{w}{\ell} \approx \frac{1}{N_\ell} \frac{2}{1/\tau+1}
\end{equation*}
Using these approximations, the infinitesimal generator for $\hat{X}_{\mathcal{R}}$ given by \eqref{eqn:QRk}--\eqref{eqn:QRl} is such that for $\mathcal{M} \subsetneq \{ 1,2,\cdots, D\}$ and $k \in \{ 1,2,\cdots, D\} \setminus \mathcal{M}$,
\begin{equation*}
    \left(Q_\mathcal{R} \right)_{\mathcal{M},\mathcal{M} \cup \{ k \}} = \sum_{\ell \in \mathcal{M}}  \mi{m}{\ell}{k} \size{m}{\ell} \fp{m}{k} \approx \sum_{\ell \in \mathcal{M}} \tau \mi{w}{\ell}{k} N_\ell \frac{1}{N_k} \frac{2}{\tau+1} = \frac{2\tau}{\tau+1} \sum_{\ell \in \mathcal{M}} \frac{N_\ell}{N_k} \mi{w}{\ell}{k},
\end{equation*}
and for $\mathcal{M} \subseteq \{ 1,2,\cdots, D\}$ such that $\mathcal{M} \neq \emptyset$ and $\ell \in \mathcal{M}$, 
\begin{equation*}
    \left(Q_\mathcal{R} \right)_{\mathcal{M},\mathcal{M} \setminus \{\ell\}} = \sum_{k \notin \mathcal{M}} \mi{w}{k}{\ell} \size{w}{k} \fp{w}{\ell} \approx \sum_{k \notin \mathcal{M}} \mi{w}{k}{\ell} N_k \frac{1}{N_\ell} \frac{2}{1/\tau+1} = \frac{2\tau}{\tau+1} \sum_{k \notin \mathcal{M}} \frac{N_k}{N_\ell} \mi{w}{k}{\ell}.
\end{equation*}
Thus, $\hat{X}_{\mathcal{R}}$ is accelerated (or decelerated) by $\frac{2\tau_0}{\tau_0+1}$ when $\tau=\tau_0$ compared to when $\tau=1$, and so
\begin{equation*}
    \frac{\hat{h}_{\{\ell\}} (\tau_0)}{\hat{h}_{\{\ell\}} (1)} \approx \frac{\tau_0+1}{2\tau_0}  \quad \text{ and } \quad \frac{\hat{g}_{\{\ell\},\{1,2,\dots,D\}} (\tau_0)}{\hat{g}_{\{\ell\},\{1,2,\dots,D\}} (1)} \approx \frac{\tau_0+1}{2\tau_0}.
\end{equation*}
Since $\frac{\fp{m}{\ell} (\tau_0)}{\fp{m}{\ell} (1)} \approx \frac{2}{\tau_0+1}$, we have from \eqref{eqn:EFTfinal} that the ratio between the expected fixation time for $\tau=\tau_0$ and for $\tau=1$ is approximately $\frac{\fp{m}{\ell} (\tau_0)}{\fp{m}{\ell} (1)}  \frac{\hat{h}_{\{\ell\}} (\tau_0)}{\hat{h}_{\{\ell\}} (1)} \approx \frac{1}{\tau_0}$, and from \eqref{eqn:CEFTfinal} that the ratio between the conditional expected time for mutant fixation for $\tau=\tau_0$ and for $\tau=1$ is approximately $\frac{\hat{g}_{\{\ell\},\{1,2,\dots,D\}} (\tau_0)}{\hat{g}_{\{\ell\},\{1,2,\dots,D\}} (1)} \approx \frac{\tau_0+1}{2\tau_0}$.

\subsection{Proofs of Theorems \ref{thm:AP}--\ref{thm:CMFPT}}

Recall that $T(\eps) = T_0 + \eps T_1$ and for $i \in \I$, $R_i (\eps) = R_{0,i} + \eps R_{1,i}$. 

\noindent \textbf{Proof of Theorem \ref{thm:AP}:}

    Consider $i \in \I$. Suppose $\abs{\J} \geq 1$. By Lemmas \ref{lem:APandMFPT} and \ref{lem:TepsInverse},
    \begin{equation}
    \label{eqn:SILi}
        \mathcal{L}_{i} (\eps) = - T^{-1} (\eps) R_i (\eps) \one_i = - \left( B^{(-1)} R_{0,i} \one_i \right) \frac{1}{\eps} - \left( B^{(-1)} R_{1,i} \one_i + B^{(0)} R_{0,i} \one_i \right) + O(\eps),
    \end{equation}
    where 
    \begin{equation}
    \label{eqn:SIB-1}
        B^{(-1)} = \Wo_0 \left(\Mo_0 T_1 \Wo_0 \right)^{-1} \Mo_0
    \end{equation}
    and
    \begin{equation*}
        B^{(0)} = - \left( I - B^{(-1)} T_1 \right) \left( \left( \Wo_0 \Mo_0 -T_0 \right)^{-1} - \Wo_0 \Mo_0 \right) \left( I - T_1 B^{(-1)}\right).
    \end{equation*}
    Since 
    \begin{equation*}
        \Mo_0 R_{0,i} = 
        \begin{tikzpicture}[baseline={-0.5ex},mymatrixenv]
        \matrix [mymatrix,inner sep=5pt] (m)  
        {
        \tilde{\Pn}_{1} & 0 & 0 & 0 \\
        0 & \ddots & 0 & \vdots \\
        0 & 0 & \tilde{\Pn}_{\numJ} & 0 \\
        };
        \draw[densely dashed] (1,-0.9) -- (1,0.9);
        \end{tikzpicture}
        \begin{tikzpicture}[baseline={0ex},mymatrixenv]
        \matrix [mymatrix,inner sep=5pt] (m)  
        {
        0 \\
        \vdots \\
        0 \\
        R_i^{(0)} \\
        };
        \draw[densely dashed] (-0.5,-0.5) -- (0.5,-0.5);
        \end{tikzpicture}
        =0
    \end{equation*}
    and by \eqref{eqn:GIforD0} and Lemma \ref{lem:AP0},
    \begin{eqnarray*}
        && \left( \left( \Wo_0 \Mo_0 -T_0 \right)^{-1} - \Wo_0 \Mo_0 \right) R_{0,i} \one_i \\
        &=& 
        \begin{tikzpicture}[baseline={0ex},mymatrixenv]
        \matrix [mymatrix,inner sep=5pt] (m)  
        {
        \tilde{D}_1 & 0 & 0 & 0 \\
        0 & \ddots & 0 & \vdots \\
        0 & 0 & \tilde{D}_{\numJ} & 0 \\
        * & \cdots & * & \left( - \tilde{T}^{(0)} \right)^{-1} \\
        };
        \draw[densely dashed] (0.5,-1.3) -- (0.5,1.3);
        \draw[densely dashed] (-2.4,-0.45) -- (2.4,-0.45);
        \end{tikzpicture}
        \begin{tikzpicture}[baseline={0ex},mymatrixenv]
        \matrix [mymatrix,inner sep=5pt] (m)  
        {
        0 \\
        \vdots \\
        0 \\
        R_i^{(0)} \\
        };
        \draw[densely dashed] (-0.5,-0.5) -- (0.5,-0.5);
        \end{tikzpicture}
        \one_i =
        \begin{tikzpicture}[baseline={0ex},mymatrixenv]
        \matrix [mymatrix,inner sep=5pt] (m)  
        {
        0 \\
        \vdots \\
        0 \\
        - \left( \tilde{T}^{(0)} \right)^{-1} R_i^{(0)} \one_i \\
        };
        \draw[densely dashed] (-1.5,-0.45) -- (1.5,-0.45);
        \end{tikzpicture} = \mathcal{L}_{0,i},
    \end{eqnarray*}
    we have
    \begin{equation*}
        B^{(-1)} R_{0,i} = \Wo_0 \left(\Mo_0 T_1 \Wo_0 \right)^{-1} \Mo_0 R_{0,i} = 0
    \end{equation*}
    and
    \begin{eqnarray*}
        B^{(0)} R_{0,i} \one_i &=& - \left( I - B^{(-1)} T_1 \right) \left( \left( \Wo_0 \Mo_0 -T_0 \right)^{-1} - \Wo_0 \Mo_0 \right) \left( I - T_1 B^{(-1)}\right) R_{0,i} \one_i \\
        &=& - \left( I - B^{(-1)} T_1 \right) \left( \left( \Wo_0 \Mo_0 -T_0 \right)^{-1} - \Wo_0 \Mo_0 \right)  R_{0,i} \one_i \\
        &=& - \left( I - B^{(-1)} T_1 \right) \mathcal{L}_{0,i}.
    \end{eqnarray*}
    Thus, \eqref{eqn:SILi} becomes
    \begin{eqnarray*}
        \mathcal{L}_{i} (\eps) &=& - 0 - \left(B^{(-1)} R_{1,i} \one_i - \left( I - B^{(-1)} T_1 \right) \mathcal{L}_{0,i} \right) + O(\eps) \\
        &=& \mathcal{L}_{0,i} - B^{(-1)} \left(R_{1,i} \one_i + T_1 \mathcal{L}_{0,i} \right) + O(\eps) \\
        &=& \mathcal{L}_{0,i} - \Wo_0 \left(\Mo_0 T_1 \Wo_0 \right)^{-1} \Mo_0 \left(R_{1,i} \one_i + T_1 \mathcal{L}_{0,i} \right) + O(\eps) \\
        &=& \mathcal{L}_{0,i} - \Wo_0 \left(\Mo_0 T_1 \Wo_0 \right)^{-1} J_i + O(\eps),
    \end{eqnarray*}
    where the second equality holds by rearranging the terms, the third equality holds by \eqref{eqn:SIB-1} and the last equality holds by the definition of $J_i$ in \eqref{eqn:Ji}.
    
    Thus, for $j \in \J$, 
    the probability for $X^\eps$ to be absorbed in $\CC_i$ starting from $y \in \tilde{C}_j$ is
    \begin{equation*}
        \mathcal{L}_{yi} (\eps) = \left( \mathcal{L}_{i} (\eps) \right)_y = 0 - e_j^T \left(\Mo_0 T_1 \Wo_0 \right)^{-1} J_i + O(\eps) = \hat{\mathcal{L}}_{ji} + O(\eps)
    \end{equation*}
    where $e_j \in \R^{\abs{\J}}$ is the standard unit basis vector for $j \in \J$, and the last equality holds by Lemma \ref{lem:MTWInvertible}. 
    Similarly, the probability for $X^\eps$ to be absorbed in $\CC_i$ starting from $y \in \tilde{\T}$ is
    \begin{eqnarray}
        \mathcal{L}_{yi} (\eps) = \left( \mathcal{L}_{i} (\eps) \right)_y &=& (L_{i})_y - \left( \sum_{j \in \J} (\tilde{L}_j)_y e_j^T \right) \left(\Mo_0 T_1 \Wo_0 \right)^{-1} J_i + O(\eps) \nonumber \\
        &=& (L_{i})_y + \sum_{j \in \J} (\tilde{L}_j)_y \hat{\mathcal{L}}_{ji} + O(\eps). \label{eqn:Lhat_Ttilde}
    \end{eqnarray}

    Suppose $\abs{\J} = 0$. Then, $\T = \tilde{\T}$, and by Lemmas \ref{lem:APandMFPT} and \ref{lem:TepsInverse},
    \begin{equation*}
        \mathcal{L}_{i} (\eps) = - T^{-1} (\eps) R_i (\eps) \one_i = - (T_0)^{-1} R_{0,i} \one_i + O(\eps) = L_{i} + O(\eps).
    \end{equation*}
    which is consistent with the formula in \eqref{eqn:Lhat_Ttilde}.

\noindent \textbf{Proof of Theorem \ref{thm:MFPT}:}

    Suppose $\abs{\J} \geq 1$. By Lemmas \ref{lem:APandMFPT} and \ref{lem:TepsInverse},
    \begin{eqnarray*}
        h (\eps) = - T^{-1} (\eps) \one_\T &=& - \frac{1}{\eps} \Wo_0 \left(\Mo_0 T_1 \Wo_0 \right)^{-1} \Mo_0 \one_\T + O(1) \\
        &=& - \frac{1}{\eps} \Wo_0 \left(\Mo_0 T_1 \Wo_0 \right)^{-1} \one_\J + O(1).
    \end{eqnarray*}
    where the last equality holds because $\tilde{\Pn}_{j} \hat{\one}_j=1$ for each $j \in \J$ and so
    \begin{equation*}
        \Mo_0 \one_\T = 
        \begin{tikzpicture}[baseline={-0.5ex},mymatrixenv]
        \matrix [mymatrix,inner sep=5pt] (m)  
        {
        \tilde{\Pn}_{1} & 0 & 0 & 0 \\
        0 & \ddots & 0 & \vdots \\
        0 & 0 & \tilde{\Pn}_{\numJ} & 0 \\
        };
        \draw[densely dashed] (1,-0.9) -- (1,0.9);
        \end{tikzpicture}
        \begin{tikzpicture}[baseline={0ex},mymatrixenv]
        \matrix [mymatrix,inner sep=5pt] (m)  
        {
        \tilde{\one}_1 \\
        \vdots \\
        \tilde{\one}_{\numJ} \\
        \one_{\tilde{\T}} \\
        };
        \draw[densely dashed] (-0.5,-0.6) -- (0.5,-0.6);
        \end{tikzpicture}
        =\one_\J.
    \end{equation*}
    
    Thus, for $j \in \J$, 
    the expected value of the first time that $X^\eps$ exits from the transient set $\T$ starting from $y \in \tilde{C}_j$ is
    \begin{equation*}
        h_y (\eps) = - \frac{1}{\eps} e_j^T \left(\Mo_0 T_1 \Wo_0 \right)^{-1} \one_\J + O(1) = \hat{h}_j \eps^{-1} + O(1),
    \end{equation*}
    where $e_j \in \R^{\abs{\J}}$ is the standard unit basis vector for $j \in \J$, and the last equality holds by Lemma \ref{lem:MTWInvertible}. 
    Similarly, the expected value of the first time that $X^\eps$ exits from the transient set $\T$ starting from $y \in \tilde{\T}$ is
    \begin{equation*}
        h_y (\eps) = - \frac{1}{\eps} \left( \sum_{j \in \J} (\tilde{L}_j)_y e_j^T \right)  \left(\Mo_0 T_1 \Wo_0 \right)^{-1} \one_\J = \left( \sum_{j \in \J} (\tilde{L}_j)_y \hat{h}_j \right) \eps^{-1} + O(1).
    \end{equation*}

    Suppose $\abs{\J} = 0$. Then, $\T = \tilde{\T}$, and by Lemmas \ref{lem:APandMFPT} and \ref{lem:TepsInverse},
    \begin{equation*}
        h (\eps) = - T^{-1} (\eps) \one_\T = - (T_0)^{-1} \one_\T + O(\eps),
    \end{equation*}
    where for $y \in \T = \tilde{\T}$, $\E_y^0 \left[ \tau^0 \right] = (- (T_0)^{-1} \one_\T)_y < \infty$ by a similar argument to the one used in  Lemma \ref{lem:APandMFPT}.

\noindent \textbf{Proof of Theorem \ref{thm:CMFPT}:}

Consider $i \in \I$. 
Suppose $\abs{\J} \geq 1$. 
We first find $\phi_i (\eps) = - T^{-1} (\eps) \mathcal{L}_{i} (\eps)$ defined in Lemma \ref{lem:APandMFPT}. By Lemma \ref{lem:TepsInverse} and Theorem \ref{thm:AP},
\begin{eqnarray*}
     \phi_i (\eps) = - T^{-1} (\eps) \mathcal{L}_{i} (\eps) &=& - \left( \frac{1}{\eps} B^{(-1)} \one_\T + O(1) \right) \left( \mathcal{L}_{0,i} - \Wo_0 \left(\Mo_0 T_1 \Wo_0 \right)^{-1} J_i + O(\eps) \right) \\
        &=& - \frac{1}{\eps} \Wo_0 \left(\Mo_0 T_1 \Wo_0 \right)^{-1} \Mo_0 \left( \mathcal{L}_{0,i} - \Wo_0 \left(\Mo_0 T_1 \Wo_0 \right)^{-1} J_i \right) + O(1) \\
        &=& - \frac{1}{\eps} \Wo_0 \left(\Mo_0 T_1 \Wo_0 \right)^{-1} \left(\Mo_0 T_1 \Wo_0 \right)^{-1} J_i + O(1).
\end{eqnarray*}
where in the third equality, we used the fact that 
\begin{equation*}
    \Mo_0 \mathcal{L}_{0,i} = 
    \begin{tikzpicture}[baseline={-0.5ex},mymatrixenv]
    \matrix [mymatrix,inner sep=5pt] (m)  
    {
    \tilde{\Pn}_{1} & 0 & 0 & 0 \\
    0 & \ddots & 0 & \vdots \\
    0 & 0 & \tilde{\Pn}_{\numJ} & 0 \\
    };
    \draw[densely dashed] (1,-0.9) -- (1,0.9);
    \end{tikzpicture}
    \begin{tikzpicture}[baseline={0ex},mymatrixenv]
    \matrix [mymatrix,inner sep=5pt] (m)  
    {
    0 \\
    \vdots \\
    0 \\
    L_i \\
    };
    \draw[densely dashed] (-0.4,-0.6) -- (0.4,-0.6);
    \end{tikzpicture}
    =0
\end{equation*}
and
\begin{equation*}
    \Mo_0 \Wo_0 = 
    \begin{tikzpicture}[baseline={-0.5ex},mymatrixenv]
    \matrix [mymatrix,inner sep=5pt] (m)  
    {
    \tilde{\Pn}_{1} & 0 & 0 & 0 \\
    0 & \ddots & 0 & \vdots \\
    0 & 0 & \tilde{\Pn}_{\numJ} & 0 \\
    };
    \draw[densely dashed] (1,-0.9) -- (1,0.9);
    \end{tikzpicture}
    \begin{tikzpicture}[baseline={0ex},mymatrixenv]
    \matrix [mymatrix,inner sep=5pt] (m)  
    {
    \tilde{\one}_{1} & 0 & 0 \\
    0 & \ddots & 0 \\
    0 & 0 & \tilde{\one}_{\numJ} \\
    \tilde{L}_{1} & \cdots & \tilde{L}_{\numJ} \\
    };
    \draw[densely dashed] (-1.4,-0.6) -- (1.4,-0.6);
    \end{tikzpicture}
    =I.
\end{equation*}
By Lemma \ref{lem:MTWInvertible}, 
we have $\hat{\mathcal{L}}_{i} = \left( \hat{\mathcal{L}}_{ji} \right)_{j \in \J} = \left(\Mo_0 T_1 \Wo_0 \right)^{-1} J_i$, and so
\begin{equation*}
    \phi_i (\eps) = - \frac{1}{\eps} \Wo_0 \left(\Mo_0 T_1 \Wo_0 \right)^{-1} \hat{\mathcal{L}}_{i} + O(1).
\end{equation*}

Thus, by Lemma \ref{lem:APandMFPT}, 
for $j \in \J$, if $\hat{\mathcal{L}}_{ji} > 0$,
the conditional expected value of the first time that $X^\eps$ exits from the transient set $\T$ given that $X^\eps$ enters $\CC_i$, starting from $y \in \tilde{C}_j$, is
\begin{equation*}
    g_{yi} (\eps) = \frac{\phi_{yi} (\eps)}{\mathcal{L}_{yi} (\eps)} = \frac{- \frac{1}{\eps} e_j^T \left(\Mo_0 T_1 \Wo_0 \right)^{-1} \hat{\mathcal{L}}_{i} + O(1)}{\hat{\mathcal{L}}_{ji} + O(\eps)} 
    = \hat{g}_{ji} \eps^{-1} + O(1).
\end{equation*}
where $e_j \in \R^{\abs{\J}}$ is the standard unit basis vector for $j \in \J$, and the last equality holds by Lemma \ref{lem:MTWInvertible}. 
Similarly, for $y \in \tilde{\T}$ such that $(L_{i})_y + \sum_{j \in \J} (\tilde{L}_j)_y \hat{\mathcal{L}}_{ji} > 0$, the conditional expected value of the first time that $X^\eps$ exits from the transient set $\T$ given that $X^\eps$ enters $\CC_i$, starting from $y$, is
\begin{eqnarray*}
    g_{yi} (\eps) = \frac{\phi_{yi} (\eps)}{\mathcal{L}_{yi} (\eps)} &=& \frac{- \frac{1}{\eps} \left( \sum_{j \in \J} (\tilde{L}_j)_y e_j^T \right)  \left(\Mo_0 T_1 \Wo_0 \right)^{-1} \hat{\mathcal{L}}_{i} + O(1)}{(L_{i})_y + \sum_{j \in \J} (\tilde{L}_j)_y \hat{\mathcal{L}}_{ji} + O(\eps)} \\
    &=& \left( \frac{\sum_{j \in \J} (\tilde{L}_j)_y \hat{\phi}_{ji}}{(L_{i})_y + \sum_{j \in \J} (\tilde{L}_j)_y \hat{\mathcal{L}}_{ji}} \right) \eps^{-1} + O(1).
\end{eqnarray*}
where $\hat{\phi}_{ji} = \left( - \left( \Mo_0 T_1 \Wo_0 \right)^{-1} \hat{\mathcal{L}}_{i} \right)_j$.

Suppose $\abs{\J} = 0$. Then, $\T = \tilde{\T}$. By Theorem 4.1,
$\mathcal{L}_{i} (\eps) = L_{i} + O(\eps)$,
and so by Lemmas \ref{lem:APandMFPT} and \ref{lem:TepsInverse},
\begin{equation*}
    \phi_i (\eps) = - T^{-1} (\eps) \mathcal{L}_{i} (\eps) = (- (T_0)^{-1} + O(\eps)) (L_{i} + O(\eps)) = - (T_0)^{-1} L_{i} + O(\eps).
\end{equation*}
Thus, for $y \in \T$ such that $(L_i)_y =\PP_y^0 [X^0(\tau^0) \in \CC_i]  > 0$, 
\begin{equation*}
    g_{yi} (\eps) = \frac{(- (T_0)^{-1} L_{i})_y + O(\eps)}{(L_i)_y + O(\eps)} = \frac{(- (T_0)^{-1} L_{i})_y}{(L_i)_y} + O(\eps)
\end{equation*}
where $\E_y^0 \left[ \tau^0 | X^0(\tau^0) \in \CC_i \right] = \frac{(- (T_0)^{-1} L_{i})_y}{(L_i)_y}< \infty$ by a similar argument to the one used in Lemma \ref{lem:APandMFPT}.

\subsection{Mutant Fixation Analysis for two-deme population}
\label{sec:2Dresults}

As a supplement to the analysis in section \ref{sec:general}, here we provide an explicit exposition on the mutant fixation analysis for fragmented population with two demes introduced in section \ref{sec:2D}. Recall that the family of continuous time Markov chains $\{ X^\eps: \eps \geq 0 \}$ has a common finite state space
\begin{equation*}
    \X = \{(w_1,m_1,w_2,m_2): 1 \leq w_1+m_1 \leq K_1, 1 \leq w_2+m_2 \leq K_2 \}.
\end{equation*}
When $\eps > 0$, the Markov chain $X^\eps$ has two recurrent classes: 
\begin{equation*}
\begin{aligned}
\CC_{ww} &= \{(w_1,0,w_2,0) \in \X: 1 \leq w_1 \leq K_1, 1 \leq w_2 \leq K_2\}, \\
\CC_{mm} &= \{(0,m_1,0,m_2) \in \X: 1 \leq m_1 \leq K_1, 1 \leq m_2 \leq K_2\},   
\end{aligned}
\end{equation*}
and the set of transient states for $X^\eps$ is $\T = \X \setminus \left( \CC_{ww} \cup \CC_{mm} \right)$. 
When $\eps = 0$, 
$\CC_{ww}$ and $\CC_{mm}$ are still two recurrent classes for $X^0$. In addition, there are two more recurrent classes for $X^0$, which are
\begin{equation*}
\begin{aligned}
\tilde{\CC}_{wm} &= \{(w_1,0,0,m_2) \in \X: 1 \leq w_1 \leq K_1, 1 \leq m_2 \leq K_2\}, \\
\tilde{\CC}_{mw} &= \{(0,m_1,w_2,0) \in \X: 1 \leq m_1 \leq K_1, 1 \leq w_2 \leq K_2\},   
\end{aligned}
\end{equation*}
and the set of transient states for $X^0$ is $\tilde{\T} = \X \setminus \left( \CC_{ww} \cup \CC_{mm} \cup \tilde{\CC}_{wm} \cup \tilde{\CC}_{mw} \right)$. 

In the absence of migration events (i.e., when $\eps = 0$), each deme evolves independently, and so the dynamics in deme $\ell$ can be described by the system presented in section \ref{sec:1D} with $K=K_\ell$.
To describe the quantities introduced in section \ref{sec:1D} for deme $\ell$, we add a superscript $\ell$ to the notations. More precisely, we use $\qss{w}{\ell} (\cdot)$ and $\qss{m}{\ell} (\cdot)$ to denote the corresponding stationary distributions for the wild-type and for the mutant in deme $\ell$, respectively, and we use $\fpone{m}{\ell} (\cdot,\cdot)$ to denote the probabilities of mutant fixation in deme $\ell$. Similarly, for $x = (n,k) \in \X$, we use $\fpone{w}{\ell}  (n,k)$ to denote the probability of wild-type fixation starting from $n$ wild-type individuals and $k$ mutant individuals in deme $\ell$.
For $\ell \in \{ 1,2 \}$,
we let 
\begin{equation*}
    \size{w}{\ell} = \sum_{n=2}^{K_\ell} n \qss{w}{\ell} (n) \qquad \text{ and } \qquad \fp{m}{\ell} = \sum_{n=1}^{K_\ell-1} \qss{w}{\ell} (n) \fpone{m}{\ell} (n,1).
\end{equation*}
This $\size{w}{\ell}$ is approximately the expected number of wild-type individuals in deme $\ell$ under its stationary distribution, provided $\qss{w}{\ell} (1)$ is small, and $\fp{m}{\ell}$ is the probability of mutant fixation when a single mutant is introduced from outside of deme $\ell$ to deme $\ell$ with a population of only the wild-type under its stationary distribution.
Similarly, for $\ell \in \{ 1,2 \}$, we let 
\begin{equation*}
    \size{m}{\ell} = \sum_{n=2}^{K_\ell} n \qss{m}{\ell} (n) \qquad \text{ and } \qquad \fp{w}{\ell} = \sum_{n=1}^{K_\ell-1} \qss{m}{\ell} (n) \fpone{w}{\ell} (1,n).
\end{equation*}

Using the construction in section \ref{sec:reducedMC}, we obtain a reduced process $\hat{X}_\mathcal{R}$, which is a continuous time Markov chain with state space $\mathcal{R} = \{ww,wm,mw,mm\}$ where the states $ww$ and $mm$ are absorbing and the states $wm$ and $mw$ are transient. We proceed to find the non-zero infinitesimal transition rates for $\hat{X}_\mathcal{R}$.

Since each deme evolves independently when $\eps=0$, the stationary distribution $\tilde{\Pn}_{wm}$ supported on $\tilde{\CC}_{wm}$ for $X^0$ is such that for $1 \leq w_1 \leq K_1, 1 \leq m_2 \leq K_2$, 
\begin{equation*}
    \tilde{\Pn}_{wm} (w_1,0,0,m_2) =  \qssnumeric{w}{1} (w_1) \qssnumeric{m}{2} (m_2),
\end{equation*}
and moreover, for $1 \leq w_1 \leq K_1 - 1, 1 \leq m_2 \leq K_2$, starting from $(w_1,1,0,m_2)$, the probability that $X^0$ will be absorbed to the recurrent class $\tilde{\CC}_{mm}$ is 
\begin{equation}
\label{eqn:Pw10m}
    \PP_{(w_1, 1, 0, m_2)}^0 \left[X^0(\tau^0) \in \CC_{mm} \right] = \fponenumeric{m}{1} (w_1,1).
\end{equation}
Thus, the infinitesimal transition rate from $wm$ to $mm$ in $\hat{X}_\mathcal{R}$ is 
\begin{eqnarray*}
    (Q_\mathcal{R})_{wm,mm} 
    &=& \sum_{\substack{x = (w_1, 0, 0, m_2): \\ 1 \leq w_1 \leq K_1 \\ 1 \leq m_2 \leq K_2}} \tilde{\Pn}_{wm} (x) Q^{(1)}_{x,x+(0, 1, 0, -1)} \PP_{x+(0, 1, 0, -1)}^0 \left[X^0(\tau^0) \in \CC_{mm} \right] \\
    &=&  \sum_{w_1=1}^{K_1-1} \sum_{m_2=2}^{K_2} \qssnumeric{w}{1} (w_1) \qssnumeric{m}{2} (m_2) \bar{\mu}_m m_2 \fponenumeric{m}{1} (w_1,1) \\
    &=& \bar{\mu}_m \sum_{w_1=1}^{K_1-1} \qssnumeric{w}{1} (w_1) \fponenumeric{m}{1} (w_1,1) \sum_{m_2=2}^{K_2} m_2 \qssnumeric{m}{2} (m_2)    \\
    &=& \bar{\mu}_m \sizenumeric{m}{2} \fpnumeric{m}{1}.
\end{eqnarray*}
This means the infinitesimal transition rate in $\hat{X}_\mathcal{R}$ to turn the first letter of the state $wm$ from a ``$w$" to a ``$m$" is proportional to the mutant migration rate, (approximately) the expected size of the mutant deme, and the probability of mutant fixation in the wild-type deme.
Similarly, the other positive infinitesimal transition rates are
\begin{equation*}
    (Q_\mathcal{R})_{mw,ww} 
    = \bar{\mu}_w \sizenumeric{w}{2} \fpnumeric{w}{1}, \qquad (Q_\mathcal{R})_{mw,mm} = \bar{\mu}_m \sizenumeric{m}{1} \fpnumeric{m}{2}, \qquad (Q_\mathcal{R})_{wm,ww} = \bar{\mu}_w \sizenumeric{w}{1} \fpnumeric{w}{2}.
\end{equation*}
With this, we will be able to find the probability of mutant fixation and the (unconditional and conditional) expected time for fixation, which we now present.

\subsubsection{Probability of mutant fixation}

We first find the probability that both demes are fixated by the mutant, starting from any deterministic state with only one mutant. 

Let $y=(w_1,1,w_2,0)$ where $1 \leq w_1 \leq K_1 - 1$ and $1 \leq w_2 \leq K_2$. We observe that $\PP_y^0 [X^0(\tau^0) \in \CC_{mm}] = \PP_y^0 [X^0(\tau^0) \in \tilde{\CC}_{wm}] = 0$, and similar to the reasoning for obtaining \eqref{eqn:Pw10m}, we have $\PP_y^0 [X^0(\tau^0) \in \tilde{\CC}_{mw}] = \fponenumeric{m}{1} (w_1,1)$. Moreover, we have $\hat{\mathcal{L}}_{mw,mm} = \frac{(Q_\mathcal{R})_{mw,mm}}{(Q_\mathcal{R})_{mw,ww}+(Q_\mathcal{R})_{mw,mm}} = \frac{\bar{\mu}_m \sizenumeric{m}{1} \fpnumeric{m}{2}}{\bar{\mu}_w \sizenumeric{w}{2} \fpnumeric{w}{1} + \bar{\mu}_m \sizenumeric{m}{1} \fpnumeric{m}{2}}$. Thus, from \eqref{eqn:mainthmAP} in Theorem \ref{thm:AP}, we have 
\begin{eqnarray}
    && \mathcal{L}_{y,mm} (\eps) \nonumber \\
    &=& \left( \PP_y^0 [X^0(\tau^0) \in \CC_{mm}] + \PP_y^0 [X^0(\tau^0) \in \tilde{\CC}_{wm}] \hat{\mathcal{L}}_{wm,mm} + \PP_y^0 [X^0(\tau^0) \in \tilde{\CC}_{mw}] \hat{\mathcal{L}}_{mw,mm}  \right) + O(\eps) \nonumber \\
    &=& \fponenumeric{m}{1} (w_1,1)  \frac{\bar{\mu}_m \sizenumeric{m}{1} \fpnumeric{m}{2}}{\bar{\mu}_w \sizenumeric{w}{2} \fpnumeric{w}{1} + \bar{\mu}_m \sizenumeric{m}{1} \fpnumeric{m}{2}} + O(\eps). \nonumber \\ \label{eqn:FP2D-1}
\end{eqnarray}
Similarly, for $z=(w_1,0,w_2,1)$ where $1 \leq w_1 \leq K_1$ and $1 \leq w_2 \leq K_2$,
\begin{equation}
\label{eqn:FP2D-2}
    \mathcal{L}_{z,mm} (\eps) 
    = \fponenumeric{m}{2} (w_2,1) \frac{\bar{\mu}_m \sizenumeric{m}{2} \fpnumeric{m}{1}}{\bar{\mu}_w \sizenumeric{w}{1} \fpnumeric{w}{2} + \bar{\mu}_m \sizenumeric{m}{2} \fpnumeric{m}{1}} + O(\eps).
\end{equation}

Assume the wild-type population has come to equilibrium, i.e., $X^\eps$ is restricted to the recurrent class $\CC_{ww}$. For $\eps \geq 0$, $X^\eps$ restricted to the recurrent class $\CC_{ww}$ is positive recurrent and has a unique stationary distribution $\Pn^\eps_{ww}$. Thus, by Theorem 4.1 in Campos et al. \cite{bib:epifinite}, 
\begin{equation}
\label{eqn:WQSS2D}
    \Pn^\eps_{ww} (w_1,0,w_2,0) = \Pn_{ww} (w_1,0,w_2,0)  + O(\eps) = \qssnumeric{w}{1} (w_1) \qssnumeric{w}{2} (w_2) + O(\eps).
\end{equation}
where $\Pn_{ww}$ is the stationary distribution supported on $\CC_{ww}$ for $X^0$. We then consider an immigration event where a mutant is introduced from an external source and the probability it falls in a particular deme is proportional to the carrying capacity of that deme. 
Then, from \eqref{eqn:FP2D-1}--\eqref{eqn:WQSS2D}, we have that the probability that both demes are fixated by the mutant is
\begin{eqnarray*}
    && \frac{K_1}{K_1+K_2} \sum_{w_1=1}^{K_1-1} \sum_{w_2=1}^{K_2} \Pn^\eps_{ww} (w_1,0,w_2,0) \mathcal{L}_{(w_1,1,w_2,0),mm} (\eps) \\
    && \qquad \qquad + \frac{K_2}{K_1+K_2} \sum_{w_1=1}^{K_1} \sum_{w_2=1}^{K_2-1} \Pn^\eps_{ww} (w_1,0,w_2,0) \mathcal{L}_{(w_1,0,w_2,1),mm} (\eps) \\
    &=& \frac{K_1}{K_1+K_2} \fpnumeric{m}{1} \frac{\bar{\mu}_m \sizenumeric{m}{1} \fpnumeric{m}{2}}{\bar{\mu}_w \sizenumeric{w}{2} \fpnumeric{w}{1} + \bar{\mu}_m \sizenumeric{m}{1} \fpnumeric{m}{2}} + \frac{K_2}{K_1+K_2} \fpnumeric{m}{2} \frac{\bar{\mu}_m \sizenumeric{m}{2} \fpnumeric{m}{1}}{\bar{\mu}_w \sizenumeric{w}{1} \fpnumeric{w}{2} + \bar{\mu}_m \sizenumeric{m}{2} \fpnumeric{m}{1}} + O(\eps).
\end{eqnarray*}

\subsubsection{Expected time for fixation}

Similarly to the previous subsection, we first find the expected time when either the mutant or the wild-type fixates, starting from any deterministic state with only one mutant. 

Let $y=(w_1,1,w_2,0)$ where $1 \leq w_1 \leq K_1 - 1$ and $1 \leq w_2 \leq K_2$. Since $\PP_y^0 [X^0(\tau^0) \in \tilde{\CC}_{wm}] = 0$,  $\PP_y^0 [X^0(\tau^0) \in \tilde{\CC}_{mw}] = \fponenumeric{m}{1} (w_1,1)$ and $\hat{h}_{mw} = \frac{1}{(Q_\mathcal{R})_{mw,ww}+(Q_\mathcal{R})_{mw,mm}} = \frac{1}{\bar{\mu}_w \sizenumeric{w}{2} \fpnumeric{w}{1} + \bar{\mu}_m \sizenumeric{m}{1} \fpnumeric{m}{2}}$, from \eqref{eqn:mainthmET} in Theorem \ref{thm:MFPT}, we have 
\begin{eqnarray}
    h_y (\eps) &=& \left( \PP_y^0 [X^0(\tau^0) \in \tilde{\CC}_{wm}] \hat{h}_{wm} + \PP_y^0 [X^0(\tau^0) \in \tilde{\CC}_{mw}] \hat{h}_{mw} \right) \eps^{-1} + O(1) \nonumber \\
    &=& \frac{\fponenumeric{m}{1} (w_1,1)}{\bar{\mu}_w \sizenumeric{w}{2} \fpnumeric{w}{1} + \bar{\mu}_m \sizenumeric{m}{1} \fpnumeric{m}{2}} \eps^{-1} + O(1). \label{eqn:MFPT2D-1}
\end{eqnarray}
Similarly, for $z=(w_1,0,w_2,1)$ where $1 \leq w_1 \leq K_1, 1 \leq w_2 \leq K_2$,
\begin{equation}
\label{eqn:MFPT2D-2}
        h_z (\eps) 
        = \frac{\fponenumeric{m}{2} (w_2,1)}{\bar{\mu}_w \sizenumeric{w}{1} \fpnumeric{w}{2} + \bar{\mu}_m \sizenumeric{m}{2} \fpnumeric{m}{1}} \eps^{-1} + O(1).
\end{equation}

Assuming the wild-type population has come to equilibrium, we then consider an immigration event where a mutant is introduced from an external source and the probability it falls in a particular deme is proportional to the carrying capacity of that deme. Thus, from \eqref{eqn:WQSS2D}--\eqref{eqn:MFPT2D-2}, we have that the expected time when either the mutant or the wild-type fixates is
\begin{eqnarray*}
    && \frac{K_1}{K_1+K_2} \sum_{w_1=1}^{K_1-1} \sum_{w_2=1}^{K_2} \Pn^\eps_{ww} (w_1,0,w_2,0) h_{(w_1,1,w_2,0)} (\eps) \\
    && \qquad \qquad + \frac{K_2}{K_1+K_2} \sum_{w_1=1}^{K_1} \sum_{w_2=1}^{K_2-1} \Pn^\eps_{ww} (w_1,0,w_2,0) h_{(w_1,0,w_2,1)} (\eps) \\
    &=& \left(\frac{K_1}{K_1+K_2} \fpnumeric{m}{1} \frac{1}{\bar{\mu}_w \sizenumeric{w}{2} \fpnumeric{w}{1} + \bar{\mu}_m \sizenumeric{m}{1} \fpnumeric{m}{2}} + \frac{K_2}{K_1+K_2} \fpnumeric{m}{2} \frac{1}{\bar{\mu}_w \sizenumeric{w}{1} \fpnumeric{w}{2} + \bar{\mu}_m \sizenumeric{m}{2} \fpnumeric{m}{1}} \right) \eps^{-1} + O(1).
\end{eqnarray*}

\subsubsection{Conditional expected time for mutant fixation}
Similarly to the previous two subsections, we first find the conditional expected time for mutant fixation, starting from any deterministic state with only one mutant. 

Let $y=(w_1,1,w_2,0)$ where $1 \leq w_1 \leq K_1 - 1$ and $1 \leq w_2 \leq K_2$. Since $\PP_y^0 [X^0(\tau^0) \in \CC_{mm}] = \PP_y^0 [X^0(\tau^0) \in \tilde{\CC}_{wm}] = 0$, from \eqref{eqn:mainthmCET} in Theorem \ref{thm:CMFPT}, we have 
\begin{eqnarray*}
    && g_{y,mm} (\eps) \\
    &=& \frac{\PP_y^0 [X^0(\tau^0) \in \tilde{\CC}_{wm}] \hat{\phi}_{wm,mm} + \PP_y^0 [X^0(\tau^0) \in \tilde{\CC}_{mw}] \hat{\phi}_{mw,mm}}{ \PP_y^0 [X^0(\tau^0) \in \CC_{mm}] + \PP_y^0 [X^0(\tau^0) \in \tilde{\CC}_{wm}] \hat{\mathcal{L}}_{wm,mm} + \PP_y^0 [X^0(\tau^0) \in \tilde{\CC}_{mw}] \hat{\mathcal{L}}_{mw,mm}} \eps^{-1} + O(1) \\
    &=& \frac{\hat{\phi}_{mw,mm}}{\hat{\mathcal{L}}_{mw,mm}} \eps^{-1} + O(1).\\
\end{eqnarray*}
where
\begin{eqnarray*}
    \hat{\phi}_{mw,mm} = \hat{\E}_{mw} \left[ \hat{\tau} \one_{ \{ \hat{X}_\mathcal{R} (\hat{\tau}) = mm \} } \right] &=& \hat{\E}_{mw} \left[ \hat{\tau} \right] \hat{\PP}_{mw} \left[\hat{X}_\mathcal{R} (\hat{\tau}) = mm \right] \\
    &=& \frac{1}{(Q_\mathcal{R})_{mw,ww}+(Q_\mathcal{R})_{mw,mm}} \hat{\mathcal{L}}_{mw,mm},
\end{eqnarray*}
and so
\begin{equation}
\label{eqn:CMFPT2D-1}
g_{y,mm} (\eps)
= \frac{1}{\bar{\mu}_w \sizenumeric{w}{2} \fpnumeric{w}{1} + \bar{\mu}_m \sizenumeric{m}{1} \fpnumeric{m}{2}} \eps^{-1} + O(1).
\end{equation}
Similarly, for $z=(w_1,0,w_2,1)$ where $1 \leq w_1 \leq K_1, 1 \leq w_2 \leq K_2$,
\begin{equation}
\label{eqn:CMFPT2D-2}
        g_{z,mm} (\eps) 
        = \frac{1}{\bar{\mu}_w \sizenumeric{w}{1} \fpnumeric{w}{2} + \bar{\mu}_m \sizenumeric{m}{2} \fpnumeric{m}{1}} \eps^{-1} + O(1).
\end{equation}

Assuming the wild-type population has come to equilibrium, we then consider an immigration event where a mutant is introduced from an external source and the probability it falls in a particular deme is proportional to the carrying capacity of that deme. Thus, from \eqref{eqn:WQSS2D} and \eqref{eqn:CMFPT2D-1}--\eqref{eqn:CMFPT2D-2}, we have that the conditional expected time for mutant fixation is
\begin{eqnarray*}
    &&
    \left( \frac{K_1}{K_1+K_2} \sum\limits_{w_1=1}^{K_1-1} \sum\limits_{w_2=1}^{K_2} \Pn^\eps_{ww} (w_1,0,w_2,0) g_{(w_1,1,w_2,0),mm} (\eps) \mathcal{L}_{(w_1,1,w_2,0),mm} (\eps) \right. \\ 
    && \qquad \qquad \left. + \frac{K_2}{K_1+K_2} \sum\limits_{w_1=1}^{K_1} \sum\limits_{w_2=1}^{K_2-1} \Pn^\eps_{ww} (w_1,0,w_2,0) g_{(w_1,0,w_2,1),mm}(\eps) \mathcal{L}_{(w_1,0,w_2,1),mm} (\eps) \right) \Bigg/ \\
    && \qquad \left( \frac{K_1}{K_1+K_2} \sum\limits_{w_1=1}^{K_1-1} \sum\limits_{w_2=1}^{K_2} \Pn^\eps_{ww} (w_1,0,w_2,0) \mathcal{L}_{(w_1,1,w_2,0),mm} (\eps) \right. \\ 
    && \left. \qquad \qquad + \frac{K_2}{K_1+K_2} \sum\limits_{w_1=1}^{K_1} \sum\limits_{w_2=1}^{K_2-1} \Pn^\eps_{ww} (w_1,0,w_2,0) \mathcal{L}_{(w_1,0,w_2,1),mm} (\eps) \right) \\ 
    &=& \left( \frac{K_1}{K_1+K_2} \fpnumeric{m}{1} \frac{\bar{\mu}_m \sizenumeric{m}{1} \fpnumeric{m}{2}}{\bar{\mu}_w \sizenumeric{w}{2} \fpnumeric{w}{1} + \bar{\mu}_m \sizenumeric{m}{1} \fpnumeric{m}{2}} \frac{1}{\bar{\mu}_w \sizenumeric{w}{2} \fpnumeric{w}{1} + \bar{\mu}_m \sizenumeric{m}{1} \fpnumeric{m}{2}} \right. \\
    && \qquad \qquad \left. + \frac{K_2}{K_1+K_2} \fpnumeric{m}{2} \frac{\bar{\mu}_m \sizenumeric{m}{2} \fpnumeric{m}{1}}{\bar{\mu}_w \sizenumeric{w}{1} \fpnumeric{w}{2} + \bar{\mu}_m \sizenumeric{m}{2} \fpnumeric{m}{1}} \frac{1}{\bar{\mu}_w \sizenumeric{w}{1} \fpnumeric{w}{2} + \bar{\mu}_m \sizenumeric{m}{2} \fpnumeric{m}{1}} \right) \eps^{-1} \Bigg/ \\
    && \qquad \left( \frac{K_1}{K_1+K_2} \fpnumeric{m}{1} \frac{\bar{\mu}_m \sizenumeric{m}{1} \fpnumeric{m}{2}}{\bar{\mu}_w \sizenumeric{w}{2} \fpnumeric{w}{1} + \bar{\mu}_m \sizenumeric{m}{1} \fpnumeric{m}{2}} + \frac{K_2}{K_1+K_2} \fpnumeric{m}{2} \frac{\bar{\mu}_m \sizenumeric{m}{2} \fpnumeric{m}{1}}{\bar{\mu}_w \sizenumeric{w}{1} \fpnumeric{w}{2} + \bar{\mu}_m \sizenumeric{m}{2} \fpnumeric{m}{1}} \right) + O(1)\\
    &=& \left( \frac{\mathcal{K}_1}{\mathcal{K}_1+\mathcal{K}_2} \frac{1}{\bar{\mu}_w \sizenumeric{w}{2} \fpnumeric{w}{1} + \bar{\mu}_m \sizenumeric{m}{1} \fpnumeric{m}{2}}  + \frac{\mathcal{K}_2}{\mathcal{K}_1+\mathcal{K}_2} \frac{1}{\bar{\mu}_w \sizenumeric{w}{1} \fpnumeric{w}{2} + \bar{\mu}_m \sizenumeric{m}{2} \fpnumeric{m}{1}} \right) \eps^{-1} + O(1).
\end{eqnarray*}
where $\frac{\mathcal{K}_1}{\mathcal{K}_2}= \frac{K_1}{K_2} \frac{\sizenumeric{m}{1}}{ \sizenumeric{m}{2}} \frac{\bar{\mu}_w \sizenumeric{w}{1} \fpnumeric{w}{2} + \bar{\mu}_m \sizenumeric{m}{2} \fpnumeric{m}{1} }{\bar{\mu}_w \sizenumeric{w}{2} \fpnumeric{w}{1} + \bar{\mu}_m \sizenumeric{m}{1} \fpnumeric{m}{2}}.$

\printbibliography[heading=bibintoc, title={References}]
\end{appendices}
\end{refsection}

\end{document}